\documentclass[10pt,reqno,notitlepage]{article}
\bibstyle{plain}

\usepackage{hyperref}
\usepackage{amssymb}
\usepackage{amsthm}
\usepackage{graphicx}
\usepackage{amsmath}
\usepackage{datetime}
\usepackage{subfigure}
\usepackage{bbm}

\usepackage{cancel}
\usepackage[usenames,dvipsnames]{color}
\usepackage[top=1in, bottom=.9in, left=1in, right=1in]{geometry}
\usepackage{enumerate}
\usepackage{mathtools}									
\mathtoolsset{showonlyrefs=true}							
\usepackage{xcolor}

\newtheorem{theorem}{Theorem}
\newtheorem{lemma}[theorem]{Lemma}								
	
\newtheorem{corollary}[theorem]{Corollary}	
	
\theoremstyle{remark}
\newtheorem{remark}{Remark}

\renewcommand{\>}{\right\rangle}
\renewcommand{\(}{\left(}				
\renewcommand{\)}{\right)}
\renewcommand{\[}{\left[}
\renewcommand{\]}{\right]}

\def\E{\mathbb{E}}			
\def\P{\mathbb{P}}										
\def\R{\mathbb{R}}
\def\K{\mathbb{K}}

\def\z{\mathbf{z}}

\def\L{\mathcal{L}}

\def\O{\mathcal{O}}
\def\U{\mathcal{U}}

\def\F{\mathcal{F}}

\def\eps{\varepsilon}

\def\sig{\sigma}

\def\lam{\lambda}

\def\Sumi{\sum_{i=1}^2}
\def\Sumk{\sum_{k=1}^2}
\def\d{\partial}

\def\ind{\mathbb{I}}

\def \tr{\rho^{(1)}}
\def \trB{\rho^{(1,B)}}

\def \abs#1{\left| #1 \right| }
\def \sign#1{\operatorname{sign}\left({#1} \right)}
\newcommand{\e}[1]{\operatorname{e}^{#1}}

\newcommand{\Exp}[1]{\operatorname{exp} \left\{ \, #1 \,\right\}}

\usepackage[normalem]{ulem}

\newcommand{\bp}{\begin{pmatrix} } 
\newcommand{\ep}{\end{pmatrix}}

\begin{document}

\title{Optimal Investment with Correlated  Stochastic Volatility Factors \footnote{Data sharing is not applicable to this article as no new data were created or analyzed in this study.}}

\author{
Maxim Bichuch \thanks{Department of Mathematics, SUNY at Buffalo, 244 Mathematics Building
Buffalo, NY 14260, USA, email: {\tt mbichuch@buffalo.edu}. Work  is partially supported by NSF grant DMS-1736414. Research is partially supported by the Acheson J. Duncan Fund for the Advancement of Research in Statistics.}
\and Jean-Pierre Fouque
\thanks{
Department of Statistics and Applied Probability,
South Hall 5504,
University of California
Santa Barbara, CA 93106
{\tt fouque@pstat.ucsb.edu}. Work  supported by NSF grant DMS-1814091.
}
}
\date{\today}
\maketitle

\begin{abstract}
The problem of portfolio allocation in the context of stocks evolving in random environments, that is with volatility and returns depending on random factors, has attracted a lot of attention. The problem of maximizing a power utility at a terminal time with only one random factor can be linearized thanks to a classical distortion transformation. In the present paper, we address the situation with several factors using a perturbation technique around the case where these factors are perfectly correlated reducing the problem to the case with a  single factor. Our proposed approximation requires to solve numerically two linear equations in lower dimension instead of a fully non-linear HJB equation. A rigorous accuracy result is derived by constructing sub- and super- solutions so that their difference is at the desired order of accuracy. We illustrate our result with a particular model for which we have explicit formulas for the approximation. In order to keep the notations as explicit as possible, we treat the case with one stock and two factors and we describe an extension to the case with two stocks and two factors.
\end{abstract}

{\bf AMS subject classification} 91G80, 60H30.\\

{\bf JEL subject classification} G11.\\

{\bf Keywords} Optimal investment, asymptotic analysis, utility maximization, stochastic volatility.\\

\section{Introduction}

The portfolio optimization problem was first introduced and studied in the continuous-time framework in \cite{Me:69, Me:71}, which provided explicit solutions on how to trade stocks and/or how to consume so as to maximize one's utility, with risky assets following the Black-Scholes-Merton model (that is, geometric Brownian motions with constant returns and constant volatilities), and when the utility function is of specific types (for instance, {Constant Relative Risk Aversion} (CRRA)). 

Stochastic volatility models have been widely studied over the last thirty years in the context of option pricing and the presence of several factors driving volatility has been well documented (see for instance \cite{fouque2000derivatives}, \cite{FoPaSiSo:11} and references therein). In  general settings, the models are intractable and often asymptotic solutions are sought, see e.g.  \cite{sircar1999stochastic}, \cite{fouque2000mean}, \cite{fouque2003multiscale}, \cite{feng2010short}.

In a general setting, \cite{kramkov2003necessary} showed existence and uniqueness of an optimal strategy using the duality approach. 
As an alternative approach, in a Markovian setting, the portfolio optimization problem with factors driving returns and volatility can be solved directly by describing it as a solution to an HJB equation with terminal condition given by the utility function.
Example of the latter approach in a portfolio optimization problem with multiscale factor models for risky assets include \cite{FoSiZa:13}, where return and volatility are driven by fast and slow factors.
Specifically, the authors heuristically derived the asymptotic approximation to the value function and the optimal strategy
for general utility functions. This analysis is complemented in \cite{FoHu:16} and in \cite{FoHu:17} in a non-Markovian context. The multiscale feature is essential to be able to consider multiple factors, because each factor requires a unique time scale.
The analysis simplifies considerably in the case of a single factor and power utilities thanks to a {\it distortion transformation} which linearizes the problem (see \cite{Za:99}, \cite{FoSiZa:13}, \cite{FoHu:17}). 

Our aim in this paper is to solve a problem with multiple factors of the same time scale. We do so by considering the case with multi factors and power utility as a perturbation problem around the case where the factors are perfectly correlated which reduces the problem to solving  linear problems. Additionally, we find a ``nearly-optimal" strategy, among all admissible strategies, without limiting them to strategies that asymptotically a-priori converge to the zeroth order strategy. The ``nearly-optimal"  strategy, if followed, produces an expected utility of the terminal wealth matching the value function at both zeroth and first order asymptotic expansion. 

The main idea of this paper is to first calculate a heuristic asymptotic expansion in the correlation parameter. Then, based on this expansion, we derive a verification result for the HJB equation, which in turn, allows us to bound the value function from above and below similar to the method used e.g. in \cite{bichuch2012asymptotic} and \cite{bichuch2019optimal}. This procedure also produces a ``nearly-optimal"  strategy, and shows that the expected utility of the terminal wealth associated with this strategy is also within the same bounds as the value function.

The rest of the paper is structured in the following way. In Section \ref{sec:model}, we study in details the case of investments in one stock and a risk-free account where the returns and volatility of the stock are driven by two factors. Our asymtotics around the case of perfect correlation between these two factors reveals a simple correction to the value function, which takes into account an imperfect correlation as well as a simple strategy which generates the corrected value function.  
A proof of this accuracy is given in Section \ref{sec:subsuper}.

In Section \ref{sec:explicit1}, we extend the model studied in \cite{Chacko} which admits explicit formulas and enables us to illustrate the accuracy of our approximation. 

Finally, to demonstrate that our approach generalizes to the case with multi assets, we consider in Section \ref{sec:twoassets} two assets driven by two factors nearly fully correlated. We also extend the model of \cite{Chacko} in that case and we discuss the difference with the models considered in \cite{EVE2018}.

\section{Models with one Stock and two Factors}\label{sec:model}
We consider a model with a stock price driven by two correlated  stochastic volatility factors:
\begin{align}
\frac{dS(t)}{S(t)} &= \mu(Z_1(t), Z_2(t))dt + \sigma(Z_1(t), Z_2(t))\,dW(t),   \label{eq:S} \\
dZ_i(t) &=	 \alpha_i(Z_i(t))\,dt +\beta_i( Z_i(t))\, dB_i(t) \quad i=1,2. \label{eq:Z}
\end{align}
The three Brownian motions $B_1,B_2,$ and $W$ are defined on a filtered probability space $(\Omega, \F, ({\F_t})_{t\ge0},\P).$
We assume that the two Brownian motions $B_1,B_2$ are correlated according to $d\<B_1, B_2\>_t =\rho_{12} dt$, and that they are correlated to the  Brownian motion $W$ according to $d\<W, B_i\>_t =\rho_i dt,~i=1,2,$  with constant coefficients $\rho_{12},\rho_1,\rho_2$ such that  $|\rho_{12}|\leq 1,|\rho_1|<1,|\rho_2|<1$ and 
\begin{align}
1+2\rho_1\rho_2\rho_{12}-\rho_1^2-\rho_2^2-\rho_{12}^2\ge 0.
\label{eq:rho-cond}
\end{align}
Throughout the paper, we work under standing classical hypotheses on the coefficients of the system \eqref{eq:S}-\eqref{eq:Z} ensuring existence and uniqueness of a strong solution.

We assume also that the market contains a bond, that carries zero interest rate for convenience. Let $\pi_t$ be the number of shares of stock held at time $t$. Thus, the evolution of the wealth process $X_t$ following the self-financing strategy $\pi_t$ is given by:
\begin{align}
d X(t) = \pi_t \frac{dS(t)}{S(t)} = \pi_t  \mu(Z_1(t), Z_2(t))dt + \pi_t \sigma(Z_1(t), Z_2(t))\,dW(t),
\label{eq:wealth}
\end{align}
and the value function of the optimal investment problem with terminal time $T$ and utility $\U$ is the following:
\begin{align}\label{eq:v}
v(t,x,z_1, z_2) = \sup_{\pi} \E_{t,x,z_1, z_2}\[\U(X_T)\],
\end{align}
where $\E_{t,x,z_1, z_2}[\cdot]$   denotes the conditional expectation $\E[\cdot \vert X(t) = x, Z_1(t)=z_1, Z_2(t)=z_2]$, and the supremum is taken over all admissible Markovian strategies $\pi_s=\pi(s, X(s), Z_1(s),Z_2(s))$ such that $X(s)$ stays nonnegative for all $t\leq s\leq T$ given $X(t)=x, Z_1(t)=z_1, Z_2(t)=z_2$, and satisfy the integrability condition
\begin{align}
\E\left\{\int_0^T \pi_t^2 \,\sigma^2(Z_1(t), Z_2(t)) dt\right\}<\infty.
\label{eq:admissible}
\end{align}
In this paper we consider the case with utility functions $\U$ being  of power type: 
\begin{align}
\U(x) = \frac{x^{p}}{p},\quad p<1, \quad p\ne 0.
\end{align}
Define the differential operators
\begin{align}
\L^{\pi, \rho_1, \rho_2}_{x,z_1, z_2} &=\pi \mu(z_1, z_2) \frac{\d}{\d x}+ \frac12 \pi^2\sig^2(z_1,z_2)  \frac{\d^2}{\d x^2} + \pi \sig(z_1, z_2) \Sumi\rho_i\beta_i(z_i) \frac{\d^2}{\d x \d z_i},\label{Lpirho1rho2}\\
\L_{z_1, z_2}^{\rho_{12}} &=\Sumi \alpha_i(z_i) \frac{\d}{\d z_i} + \frac12 \Sumi \beta_i^2(z_i)  \frac{\d^2}{\d z_i^2} + \rho_{12}\beta_1(z_1)\beta_2(z_2) \frac{\d^2}{\d z_1 \d z_2}.\label{Lrho12}
\end{align}
The value function $v$ satisfies:
\begin{align}
&\d_t v + \L_{z_1, z_2}^{\rho_{12}}  v + \sup_{\pi} \L^{\pi, \rho_1, \rho_2}_{x,z_1, z_2} v =0,\label{eq:HJBfull}\\
& v(T, x, z_1, z_2) = \U(x).
\end{align}
A related problem with power utility of consumption and several factors is studied in  \cite{sheu2018} where it is proved that the associated HJB equation admits a classical solution.
Maximizating over $\pi$ gives:
\begin{align}
\pi^{*} = -\frac{\mu(z_1, z_2)}{\sigma^2(z_1, z_2)}\frac{v_x}{v_{xx}} - \frac{1 }{\sig(z_1, z_2)}\frac{\Sumi \rho_i\beta_i(z_i) v_{xi}}{v_{xx}},
\label{eq:pi}
\end{align}
where $v_{i}$ denotes a derivative with respect to $z_i,~i=1,2.$
Substituting \eqref{eq:pi} into \eqref{eq:HJBfull}, it follows that 
\begin{align}
\d_t v + \L_{z_1, z_2}^{\rho_{12}} v -\frac{\(\lam(z_1, z_2)v_x + \Sumi \rho_i\beta_i(z_i) v_{xi}\)^2}{2 v_{xx}}=0.
\label{eq:HJB}
\end{align}
where the Sharpe ratio $\lambda$ is defined by $\lam(z_1, z_2) = \frac{\mu(z_1, z_2)}{\sigma(z_1, z_2)}.$

{We proceed in the next section to solve the problem when the two factors are perfectly correlated. It turns out that this solution follows  \cite{FoHu:16}. We then compute the first order perturbation adjustment, around the perfectly correlated case. In Section \ref{sec:subsuper}, using these zero and first order perturbations, we construct sub- and super-solutions to the original PDE \eqref{eq:HJB}, and rigorously show the error of the constructed approximation.
}

\subsection{Transformation of the Non-linear HJB Equation}
We perform a distortion transformation of the HJB equation \eqref{eq:HJB} for the value  function $v$, as follows. Fix $q\in\R$ and consider 
\begin{align}\label{eq:distortion}
v(t,x, z_1, z_2) = \frac{x^{p}}{p} \(\Psi(t,z_1, z_2)\)^q.
\end{align}
Then, thanks to the power utility scaling, $\Psi$ must satisfy:
\begin{align}
&\d_t \Psi + \L_{z_1, z_2}^{\rho_{12}} \Psi +\frac \Gamma{2q} \lam^2(z_1, z_2) \Psi +\Gamma   \lam(z_1, z_2) \Sumi \rho_i\beta_i(z_i)\Psi_i \label{eq:HJB-non-lin}\\
&+ \frac{1}{2\Psi}\left(
\Sumi\( (q-1)+q \rho_i^2\Gamma\)\beta_i^2(z_i) \Psi_i^2     + \beta_1(z_1)\beta_2\(z_2)(\rho_{12} (q-1)+q \rho_1\rho_2\Gamma\)  \Psi_1\Psi_2\right)
=0,\\
&\Psi(T, z_1, z_2)=1,
\end{align}
where
we denote 
\begin{align}\label{defGamma}
\Gamma = \frac{p}{1-p},\quad \mbox{so that}\quad \Gamma >-1 \quad \mbox{and}\quad \Gamma\neq 0, \quad\mbox{since} \quad p<1,\, p\neq 0.
\end{align}
We have one degree of freedom, namely $q$, and 
in order to cancel the non-linear terms, one must have
\begin{align}
(q-1)+q \rho_i^2\Gamma=0,\, i=1,2,\quad \mbox{and}\quad \rho_{12} (q-1)+q \rho_1\rho_2\Gamma=0,
\end{align}
which can be achieved only if $\rho_{12}=1$ and $\rho_1=\rho_2$, that is when the two factors $Z_1$ and $Z_2$ are perfectly correlated.
So we digress a little to review that case.

\subsection{Fully Correlated  Factors}
We start by recalling the result from \cite{FoSiZa:13} as applied to our case. More specifically, in the case of fully correlated factors $Z_1, Z_2$ we are able to easily adapt the computations there as follows.
Let us temporarily assume that $\rho_{12}=1$, then, condition \eqref{eq:rho-cond} forces us to also assume that $\rho_1=\rho_2=\rho$, with $\abs{\rho}<1$.
We consider the ``distortion transformation" used in \cite{Za:99} and \cite{FoSiZa:13}:
\begin{align}
v(t,x, z_1, z_2) = \frac{x^{p}}{p} \(\Psi^{(0)}(t,z_1, z_2)\)^q,
\label{eq:subst1}
\end{align}
where the superscript $(0)$ indicates that this function will be the zeroth order in the asymptotics presented in the following section.
 The function $\Psi^{(0)}$ satisfies
\begin{align}
&\d_t \Psi^{(0)} + \L_{z_1, z_2}^1 \Psi^{(0)} +\frac \Gamma{2q} \lam^2(z_1, z_2) \Psi^{(0)} +\Gamma  \rho \lam(z_1, z_2) \Sumi \beta_i(z_i)\Psi_i^{(0)} \\
&+ \frac12\( (q-1)+q \rho^2\Gamma\) \frac{ \(\Sumi \beta_i(z_i)\Psi_i^{(0)}\)^2}{\Psi^{(0)}}=0,
\label{eq:non-lin}
\end{align}
where $\Gamma$ is given by \eqref{defGamma}.
Choosing 
\begin{align}\label{def:q}
q= \frac{1}{1+\Gamma \rho^2}, \quad \mbox{so that}\quad 0<q<\frac{1}{1-\rho^2}\, ,
\quad \mbox{and}\quad q=1 \quad \mbox{if}\quad \rho=0,
\end{align}
 the equation for $\Psi^{(0)}$ becomes linear:
\begin{align}
\d_t \Psi^{(0)}  + \L_{z_1, z_2}^{1,\rho}  \Psi^{(0)}+\frac \Gamma{2q} \lam^2(z_1, z_2) \Psi^{(0)}
&=0,\label{eq:Psi0}\\
\Psi^{(0)}(T, z_1, z_2)&=1,
\end{align}
where
\begin{align}
\L_{z_1, z_2}^{1,\rho} \phi =  \L_{z_1, z_2}^{1} \phi +\Gamma  \rho \lam(z_1, z_2) \Sumi \beta_i(z_i)\phi_i,
\label{eq:L}
\end{align}
and $\L_{z_1, z_2}^{1}$ is given by \eqref{Lrho12} with $\rho_{12}=1$.
Note that in this case, we may assume that $B_1(t) = B_2(t) = B(t)$, and we get a Feynman--Kac type formula:
\begin{align}
\Psi^{(0)}(t, z_1, z_2) = \tilde \E_{t,z_1,z_2}\[\e{\frac \Gamma{2q} \int_t^T \lam^2(Z_1(s), Z_2(s))ds } \],\label{eq:Psi0-E} 
\end{align}
where $\tilde \P$ is defined so that 
\begin{align}
\tilde B(t)  = B(t) -\Gamma  \rho \int_0^t \lam(Z_1(s), Z_2(s))ds 
\label{eq:drift-change}
\end{align}
is a standard Brownian motion under it.

\subsection{Asymptotics Around the Fully Correlated Case}\label{sec:full-corr}

We now go back to the general correlation structure \eqref{eq:rho-cond} and the non-linear HJB equation \eqref{eq:HJB-non-lin}. Our goal is to  expand  around the fully correlated case  when $\rho_{12} =1$, and $\rho_1 = \rho_2$, presented in the previous section. 
Accordingly, we now assume that $\rho_1, \rho_2, \rho_{12}$ have the following form:
\begin{align}
\rho_i = \rho +  \rho_i^{(1)}\eps, ~i=1,2 \mbox{ and } \rho_{12} = 1 +   \rho_{12}^{(1)} \eps,
\label{rhoeps1}
\end{align}
where $\rho_{12}^{(1)}<0$ and
$\eps$ is a small parameter,  $0<\eps\ll 1$, small enough to ensure a proper covariance structure satisfying \eqref{eq:rho-cond}. Indeed, \eqref{eq:rho-cond} then becomes:
\begin{align}
1+2\rho_1\rho_2\rho_{12}-\rho_1^2-\rho_2^2-\rho_{12}^2 = 2(\rho^2-1) \rho_{12}^{(1)} \eps +O(\eps^2) \ge0,
\end{align}
for $0<\eps\ll 1$, small enough.
Consider the ansatz
\begin{align}
v(t,x, z_1, z_2) = \frac{x^{p}}{p} \( \Psi^{(0)} (t,z_1, z_2)+\eps\Psi^{(1)} (t,z_1, z_2) +\O(\eps^{2}) \)^q 
,
\label{eq:ansatz-stoch-vol}
\end{align}
where 
the exponent $q$ is given by \eqref{def:q}: $q= \frac{1}{1+\Gamma \rho^2}=\frac{1-p}{1-p(1-\rho^2)}$.
Plugging this ansatz in the HJB equation \eqref{eq:HJB} and canceling terms of zero order in $\eps$ gives that the function $\Psi^{(0)}$ satisfies \eqref{eq:Psi0} and, therefore, is given by \eqref{eq:Psi0-E}. 
Cancelling the terms of order one in $\eps$, we deduce that the function $\Psi^{(1)}$ must satisfy:
%
\begin{align}
&\d_t\Psi^{(1)} + \L _{z_1, z_2}^{1,\rho} \Psi^{(1)} + \frac{\Gamma}{2q}\lam^2 \Psi^{(1)}  +  f_1(\Psi^{(0)}, \nabla \Psi^{(0)}, \mathbb H(\Psi^{(0)}) )=0,
\label{eq:Psi1}\\
&\Psi^{(1)}(T, z_1, z_2)=0,
\end{align}
where the operator $\L_{z_1, z_2}^{1,\rho}$ is given by \eqref{eq:L}, and 
\begin{align}
f_1(\Psi^{(0)}, \nabla \Psi^{(0)}, \mathbb H(\Psi^{(0)}) )=&
\frac{q\Gamma \rho}{\Psi^{(0)}}  \left(\beta_1 \beta_2 \Psi^{(0)}_1 \Psi^{(0)}_2(\tr_{1}-\rho  \tr_{12}+\tr_{2})+\Sumi \beta_i^2 \tr_{i} (\Psi^{(0)}_i)^2\right)\\
& +\Gamma  \lam\Sumi\tr_i\beta_i  \Psi^{(0)}_i + \tr_{12}\beta_1\beta_2\Psi^{(0)}_{12}.
\end{align}
The probabilistic representation of $\Psi^{(1)}$ is given by the Feynman--Kac  type formula:
\begin{align}
\Psi^{(1)}(t, z_1, z_2) = \tilde \E_{t,z_1,z_2}\[\int_t^T \e{\frac \Gamma{2q} \int_t^s \lam^2(Z_1(\tau), Z_2(\tau))d\tau }  f_1(\Psi^{(0)}, \nabla \Psi^{(0)}, \mathbb H(\Psi^{(0)}) )(Z_1(s), Z_2(s))ds \],\label{eq:Psi1-E} 
\end{align}
under the same probability measure  $\tilde \P$ as in \eqref{eq:Psi0-E}.

We now consider a zeroth order approximation to $\pi^{*}$ given in \eqref{eq:pi}, by substituting the zeroth order approximation for $v$ from \eqref{eq:ansatz-stoch-vol}, namely, $v(t,x, z_1, z_2) \approx \frac{x^{p}}{p} \(\Psi^{(0)}(t,z_1, z_2)\)^q$, and by using $\rho_{12} =1, \rho_1 =\rho_2 =\rho,$ the zeroth order approximation from \eqref{rhoeps1}. We obtain 
\begin{align}
\pi^{0} =\frac{x}{(1-p)\sigma}\left( \lambda + \rho q\,\frac{\Sumi \beta_i \Psi^{(0)}_i}{ \Psi^{(0)}}\right).
\label{eq:pi0}
\end{align}
{Note that $X^{(\pi^{0})} >0$, and therefore once we show the appropriate integrability conditions in Corollary \ref{corr:admissible}, it will follows that $\pi^{0}$ is an admissible strategy.}

Next, we consider the value 
\begin{align}
v^{(\pi^0)}(t,x,z_1, z_2) = \E_{t,x,z_1, z_2}\[\U(X^{(\pi^0)}_T)\],
\end{align}
obtained by following the strategy $\pi^{0}$ in \eqref{eq:wealth}. It satisfies the linear equation:
\begin{align}
&\d_t v^{(\pi^0)} + \L_{z_1, z_2}^{\rho_{12}}  v^{(\pi^0)}  + \L^{\pi^0, \rho_1, \rho_2}_{x,z_1, z_2} v^{(\pi^0)} =0,\label{eq:v0}\\
& v^{(\pi^0)} (T, x, z_1, z_2) = \U(x).
\end{align}
Consistent with the previous distortion transformation \eqref{eq:distortion} letting
\begin{align}
v^{(\pi^0)}(t,x, z_1, z_2) = \frac{x^{p}}{p} \( \Psi^{(\pi^0)} (t,z_1, z_2)\)^q,\label{eq:v0expand}
\end{align}
it follows that $\Psi^{(\pi^0)}$ 
solves:
\begin{align}
&\frac{1-q}{2} \Psi^{(0)} \(\beta_1^2 \(\Psi^{(\pi^0)}_1\)^2 +\beta_1\beta_2\rho_{12} \Psi^{(\pi^0)}_1\Psi^{(\pi^0)}_2 + \beta_2^2 \(\Psi^{(\pi^0)}_2\)^2   \)+\frac{\Gamma}{2q} \Psi^2\( q^2 \rho^2 \(\Sumi \beta_i \Psi^{(0)}_i\)^2 -\lam^2 (\Psi^{(0)})^2\)\\
&-\Psi^{(\pi^0)} \Psi^{(0)} \( \d_t \Psi^{(\pi^0)} + \L_{z_1, z_2}^{\rho_{12}} \Psi^{(\pi^0)} +\frac \Gamma{2q} \lam^2(z_1, z_2) \Psi^{(\pi^0)} +\Gamma   \lam(z_1, z_2) \Sumi \rho_i\beta_i(z_i)\Psi_i^{\pi^0} \)=0,\\
&\Psi^{(\pi^0)}(T, z_1, z_2)=1.
\end{align}
A classical regular expansion argument for PDEs (as in \cite{FoSiZa:13}[Section 6.3.2] for instance) shows that 
\begin{align}
 \Psi^{(\pi^0)}  = \Psi^{(0)}+\eps\Psi^{(1)} +\O(\eps^{2}),
\end{align}
where the function $ \Psi^{(0)}$ and $ \Psi^{(1)}$ are exactly those obtained in the previous section in \eqref{eq:Psi0} and \eqref{eq:Psi1} respectively. Therefore, up to the first order in $\eps$, $v^{(\pi^0)}$ is identical to $v$ expanded heuristically in \eqref{eq:ansatz-stoch-vol}. Once we prove in Section \ref{sec:subsuper} that the expansion \eqref{eq:ansatz-stoch-vol} for $v$ is accurate, we will also be able to conclude that
the strategy $\pi^{0}$ given by \eqref{eq:pi0} generates up to order $\eps$ the value $v$ given by \eqref{eq:v} or \eqref{eq:HJBfull}.

\section{Proof of Accuracy}\label{sec:subsuper}

We now go back to the general case as in Section \ref{sec:model}. The goal is to make rigorous the previous heuristic results. In other words, we prove that the expansion in \eqref{eq:ansatz-stoch-vol} is correct. Moreover, as explained at the end of Section \ref{sec:full-corr}, we justify that the zeroth order strategy $\pi^{0}$ from \eqref{eq:pi0} indeed, achieves the maximum value up to order $\eps$.

Recall the original HJB equation \eqref{eq:HJB} for the value  function $v$, the distortion transformation \eqref{eq:distortion} and the resulting non-linear HJB equation for $\Psi$ \eqref{eq:HJB-non-lin}.

Note that we still assume that $q$ is given by \eqref{def:q}, however, \eqref{eq:HJB-non-lin} the equation for $\Psi$ remains fully nonlinear. The distortion transformation \eqref{eq:distortion} will be key to build sub- and super-solutions for \eqref{eq:HJB}, but first, we need some smoothness properties for the functions $\Psi^{(0)}$ and $\Psi^{(1)}$. In this section we will commonly use the notation that a function $f$ is bounded away from zero, which we define as $f$ is such that $\inf_{x\in \mbox{ domain of }  f} \abs{f(x)} >0$.   

\subsection{Smoothness of $\Psi^{(0)}$ and $\Psi^{(1)}$}\label{sec:smoothness}

We have the following:
\begin{lemma}\label{lem:smoothness}
Assume that  $\lam, \alpha_i, \beta_i,~i=1, 2$ are bounded, twice differentiable with bounded derivatives, and that $\sig$, 
$ \beta_i,~i=1, 2$ are bounded away from zero. 
Then, $\Psi^{(0)}$ and $\Psi^{(1)}$, the solutions of \eqref{eq:Psi0} and \eqref{eq:Psi1} respectively, exist and they are unique and bounded. Moreover, their derivatives up to order two are bounded. {Additionally, $\Psi^{(0)}$ and $\Psi^{(1)}$ are also given by their Feynman--Kac representations \eqref{eq:Psi0-E} and \eqref{eq:Psi1-E}  respectively. } 
\end{lemma}
\begin{proof}
We show the proof for $\Psi^{(0)}$, whereas the proof for $\Psi^{(1)}$ is similar.

First, note that under our coefficient assumptions, the operator $\L_{z_1,z_2}^{1,\rho}$ appearing in \eqref{eq:L}  is (degenerate) elliptic. Then, existence and uniqueness of the classical solution $\Psi^{(0)}$ of \eqref{eq:Psi0} follows from \cite{oleinik1965smoothness}[Theorem 6]. 
{Therefore, it is easily seen that all the assumptions of Feynman--Kac formula in \cite{karatzas1998brownian}[Theorem 5.7.6] 
hold. Thus, from \eqref{eq:Psi0-E}, 
it follows  that $\Psi^{(0)}$ is bounded.}

Since $\Psi^{(0)}$ is a classical solution to \eqref{eq:Psi0}, it is differentiable, and we can consider $\Psi^{(0)}_i$, its derivative with respect to $z_i, i=1,2$. 
By differentiating \eqref{eq:Psi0}, we obtain the system of PDEs:
\begin{align}
&\d_t \Psi^{(0)}_i  +\frac \Gamma{2q} \lam^2(z_1, z_2) \Psi^{(0)}_i+ \L_{z_1,z_2}^{1,\rho} \Psi^{(0)}_i\\
&+\(\alpha_i'(z_i) + \beta_i(z_i)\beta_i'(z_i)  \d_{z_i}  + \beta_i'(z_i)\beta_j(z_j) \d_{z_j}+\Gamma  \rho \(\lam_i(z_1, z_2) \beta_i(z_i)  + \lam(z_1, z_2) \beta_i'(z_i)\) \) \Psi^{(0)}_i\\
&= - \frac \Gamma{q} \lam(z_1, z_2) \lam_i(z_1, z_2) \Psi^{(0)} - \Gamma  \rho \lam_i(z_1, z_2) \beta_j(z_j) \Psi^{(0)}_j
,\label{eq:Psi0_1}\\
&\Psi^{(0)}_i(T, z_1, z_2)=0,~i,j=1,2~i\ne j.
\end{align}
Note that $\lam_i$ here, as per our convention, denotes the partial derivative of $\lam$ with respect to $z_i,~i=1,2.$ Denoting by $D\Psi^{(0)}$ the vector $(\Psi^{(0)}_1, \Psi^{(0)}_2)^T$ and by $D\lam$ the vector$(\lam_1,\lam_2)^T$, the system of equations \eqref{eq:Psi0_1} can be rewritten:
\begin{align}
&\d_t D\Psi^{(0)}+\L _{z_1,z_2}^{1,\rho}I_{2\times 2}D\Psi^{(0)}+V(z_1,z_2)D\Psi^{(0)}+\frac \Gamma{q} \lam(z_1, z_2)\Psi^{(0)}D\lam=0,
&D\Psi^{(0)}(T, z_1, z_2)=0,
\end{align}
where $ I_{2\times 2}$ is the $2\times 2$ identity matrix, $V$ is a $2\times 2$ potential matrix, and the last term being a source term.

Therefore, the assumptions of \cite{karatzas1998brownian}[Theorem 5.7.6]  again hold, and $D\Psi^{(0)}$ is given by the Feynman-Kac formula, 
\begin{equation*}
D\Psi^{(0)} (t,z_1,z_2)=\tilde \E_{t,z_1,z_2}\left[ \int_t^Te^{\int_t^sV(Z_1(u),Z_2(u))du}\( \frac \Gamma{q} \lam\Psi^{(0)}D\lam \) (Z_1(s), Z_2(s))ds \right].
\end{equation*}

Under our coefficient assumptions, this shows that $\Psi^{(0)}_1$ and $\Psi^{(0)}_2$ are bounded. Differentiating the system  \eqref{eq:Psi0_1} with respect to $z_i, i=1,2$, one obtains equations for the second order derivatives $\Psi^{(0)}_{11}, \Psi^{(0)}_{12}, \Psi^{(0)}_{22}$ and their boundedness is derived by using again a Feynman--Kac representation and our coefficient assumptions.  Here, we omit these straightforward lengthy details as well as the calculation details for $\Psi^{(1)}$ given by \eqref{eq:Psi1} and its derivatives. {Finally, we similarly conclude that the Feynman--Kac representation  \eqref{eq:Psi1-E} of  $\Psi^{(1)}$  holds.}

\end{proof}
\begin{corollary}\label{corr:admissible}
Under the assumptions of Lemma \ref{lem:smoothness}, the strategy $\pi^{0}$ given in \eqref{eq:pi0} is admissible.
\end{corollary}

\begin{proof}
Under our assumptions from \eqref{eq:Psi0-E} , we have that $\Psi^{(0)}$ is bounded away from zero. Moreover, from Lemma \ref{lem:smoothness}, we have that $\Psi^{(0)}_i,~i=1,2$ are bounded. Therefore it follows that $\frac{1}{(1-p)\sigma}\left( \lambda + \rho q\,\frac{\Sumi \beta_i \Psi^{(0)}_i}{ \Psi^{(0)}}\right)$ is also bounded. 
Thus from the definition of $\pi^0$ in \eqref{eq:pi0}, it follows that $X^{(\pi^{0})}$ given by \eqref{eq:wealth} is a generalized geometric Brownian motion, and thus is positive. Additionally, $\pi^{0}$ satisfies the admissibility constraint \eqref{eq:admissible}.
\end{proof}

\subsection{Building Sub- and Super-Solutions}\label{sec:subsuper2}

The goal is now to  obtain bounds for the value function $v$, solution to  the HJB equation \eqref{eq:HJB}, and to justify the approximation \eqref{eq:ansatz-stoch-vol}.
Consider $\Psi^{(0)}$ and $\Psi^{(1)}$ given as  solutions to \eqref{eq:Psi0} and  \eqref{eq:Psi1} respectively and under the assumptions of Lemma \ref{lem:smoothness}. Using those and the distortion transformation \eqref{eq:distortion}, define
\begin{align}
v^\pm (t,x, z_1, z_2) = \frac{x^{p}}{p} \( \Psi^{(0)} (t,z_1, z_2) + \eps \Psi^{(1)} (t,z_1, z_2)\pm \eps^2M(T-t) \)^q ,
\label{eq:sub-super}
\end{align}
where $M>0$ is a constant to be determined later independently of $\eps$, and where $q$ is given by \eqref{def:q}. Here, we assume $p<0$ to start with and the case $0<p<1$ will be treated in Section \ref{sec:pnegative}.

Observe that from the boundary conditions of  $\Psi^{(0)}$ and $\Psi^{(1)}$, we have
$
 v^\pm(T, x, z_1, z_2) = \U(x).
$
Note also that from the Feynman--Kac formula \eqref{eq:Psi0-E}, the function $\Psi^{(0)}$ is bounded, positive, and bounded away from zero. On the other hand, the function $\Psi^{(1)}$ is bounded, and, therefore, for $\eps$ small enough, $\Psi^{(0)} (t,z_1, z_2) + \eps \Psi^{(1)} (t,z_1, z_2) \pm\eps^2M(T-t)>0$, and consequently, $v^\pm$ is well defined.


\subsubsection{Strategy of the proof of accuracy}\label{sec:proofstrategy}
Recall the HJB equation \eqref{eq:HJBfull} and its two operators $\L_{z_1, z_2}^{\rho_{12}} $ and $\L^{\pi, \rho_1, \rho_2}$ defined in \eqref{Lpirho1rho2} and \eqref{Lrho12} respectively. 
From it, we define the operator $Q^\pi$ 
\begin{align}
Q^\pi=\d_t  + \L_{z_1, z_2}^{\rho_{12}}  + \L^{\pi,\rho_1,\rho_2}_{x,z_1, z_2},
\end{align}
where $\rho_{12} = 1 +  \rho_{12}^{(1)}\eps, ~\rho_i = \rho +  \rho_{i}^{(1)}\eps$.
We will show that there exists $M$ such that for $\eps$ small enough we have
\begin{align}\label{Qestimates}
Q^{\pi^0}[v^+]\geq 0,\,\,\mbox{and}\quad
\sup_{\pi}Q^{\pi}[v^-]\leq 0,
\end{align}
where the strategy $\pi^0$ is given by \eqref{eq:pi0} and the strategy $\pi$ is any admissible strategy. 
By It\^o's formula and a justification of the martingale property which will be given later, we then conclude that
\begin{align}
v(t,x,z_1,z_2)&\geq \E_{t,x,z_1,z_2}\left[\U(X^{(\pi^0)}(T))\right]=\E_{t,x,z_1,z_2}\left[v^+(T,X^{(\pi^0)}(T),Z_1(T),Z_2(T))\right]\\
&=v^+(t,x,z_1,z_2)+\E_{t,x,z_1,z_2}\left[\int_t^TQ^{\pi^0}[v^+](s,X^{(\pi^0)}(s),Z_1(s),Z_2(s))ds\right]\\
&\geq v^+(t,x,z_1,z_2),\label{subestimate}
\end{align}
\begin{align}
 \E_{t,x,z_1,z_2}\left[\U(X^{(\pi)}(T))\right]&=\E_{t,x,z_1,z_2}\left[v^-(T,X^{(\pi)}(T),Z_1(T),Z_2(T))\right]\\
&=v^-(t,x,z_1,z_2)+\E_{t,x,z_1,z_2}\left[\int_t^TQ^{\pi}[v^-](s,X^{(\pi)}(s),Z_1(s),Z_2(s))ds\right]\\
&\leq v^-(t,x,z_1,z_2)+\E_{t,x,z_1,z_2}\left[\int_t^T\sup_\pi Q^{\pi}[v^-](s,X^{(\pi)}(s),Z_1(s),Z_2(s))ds\right]\\
&\leq v^-(t,x,z_1,z_2),\label{superestimatepi}
\end{align}
and, by taking a supremum over $\pi$:
\begin{align}
v(t,x,z_1,z_2)=\sup_{\pi} \E_{t,x,z_1,z_2}\left[\U(X^{(\pi)}(T))\right]\leq v^-(t,x,z_1,z_2).
\label{superestimate}
\end{align}

In other words,  $v^+$ is a submartingale along $\pi^0$ and 
$v^-$  is a supermartingale along any admissible $\pi$. In turn,
\eqref{subestimate} and \eqref{superestimate} show that $v^-$ is a sub-solution and $v^+$ is a super-solution. Using again the definition \eqref{eq:distortion} of $v^\pm$,  we deduce that our proposed approximation is accurate at the order $\eps$:
\begin{align}\label{approx}
\left|v- \frac{x^{p}}{p} \( \Psi^{(0)} + \eps \Psi^{(1)} \)^q \right|=x^p\O(\eps^2),
\end{align}
uniformly in $(t,z_1,z_2)$.
This is formalized in the following:

\begin{theorem}\label{TheTheorem}
In addition to the coefficient assumptions in Lemma \ref{lem:smoothness}, we assume that 
$\lam$ is bounded and bounded away from zero, and $p<0$.  Then, there exits a constant $M>0$ such that, for $\eps$ small enough, the functions $v^\pm$  defined in \eqref{eq:sub-super} are super- and sub-solutions, and the accuracy of approximation \eqref{approx} holds. Moreover, the strategy  $\pi^{0}$  given by \eqref{eq:pi0}, is ``nearly-optimal'', in other words, if followed, then the expected utility of the terminal wealth will differ from the value function by $\O(\eps^2)$, i.e.
\begin{align}\label{pi0optimal}
0\leq v(t,x,z_1, z_2)-\E_{t,x,z_1,z_2}\[ \frac{1}{p}\left(X_T^{(\pi^{0}) }\right)^p\] =x^p\O(\eps^2),
\end{align}
uniformly in $(t,z_1,z_2)$. 
\end{theorem}

\begin{proof}
The proof follows the argument presented at the begining of Section \ref{sec:proofstrategy} and will mainly consists in deriving 
the key inequalities \eqref{Qestimates}.
Recall that $\rho_{12} = 1 +  \rho_{12}^{(1)}\eps, ~\rho_i = \rho +  \rho_{i}^{(1)}\eps$, and that the strategy $\pi^{0}$ is given by \eqref{eq:pi0}.

\subsubsection{Super-solution, computation of $Q^{\pi^0}[v^+]$}\label{sec:subsol}
By direct computation, we get:
\begin{align}
&\frac{Q^{\pi^0}[v^+]}{\( \Psi^{(0)}  + \eps \Psi^{(1)} +\eps^2M(T-t) \)^{q-2} }
=
q\(\frac{x^p}{p}\( \Psi^{(0)}  + \eps \Psi^{(1)} +\eps^2M(T-t) \)  \)
 \(\d_t \Psi^{(0)}  +\frac \Gamma{2q} \lam^2\Psi^{(0)}+ \L ^{1,\rho}_{z_1,z_2}\Psi^{(0)}\) \\
&\qquad+\eps q\(\frac{x^p}{p}\( \Psi^{(0)}  + \eps \Psi^{(1)} + \eps^2M(T-t) \) \)  \(\d_t\Psi^{(1)} + \L^{1,\rho}_{z_1,z_2}  \Psi^{(1)} + \frac{\Gamma}{2q}\lam^2 \Psi^{(1)}  -  f_1(\Psi^{(0)}, \nabla \Psi^{(0)}, \mathbb H(\Psi^{(0)}) )\)\\
& \qquad
 -\eps^2q\(\frac{x^p}{p}\Psi^{(0)} \)
 \left(M\left[1-\frac{\Gamma}{2q}\lambda^2(T-t)\right]-
 \frac{\Phi_{20}}{2(1-p)  (\Psi^{(0)}) ^2}\right) + \eps^3  \(\Phi_{30}^{+} + M\Phi_{31}^{+}\) + \eps^4 M\( \Phi_{41}^{+} + M\Phi_{42}^{+}\),
   \label{eq:v-}
\end{align}
where 
the quantities $\Phi_{20}, \Phi_{30}^{+}, \Phi_{31}^{+}, \Phi_{41}^{+}, \Phi_{42}^{+}$ are given by
\begin{align}
\Phi_{20}&=
 (\Psi^{(0)})^2 \Bigg[
 \sum_{i=1}^2\beta_i^2\left(qp(\rho_i^{(1)})^2(\Psi^{(0)}_i)^2
 +4 q  p \rho  \rho_i^{(1)} \Psi^{(0)}_i  \Psi^{(1)}_i -(q -1) (p-1) (\Psi^{(1)}_i)^2\right)\\
 &\hskip 3cm +2 \beta_1 \beta_2
 \bigg(
 [q  p \rho  \sum_{i=1}^2\rho_i^{(1)}-(q -1) (p-1) \rho_{12}^{(1)}](\Psi^{(1)}_1  \Psi^{(0)}_2+\Psi^{(1)}_2  \Psi^{(0)}_1)\\
& \hskip 3cm\qquad\qquad -(q-1)(p-1)\Psi^{(1)}_1  \Psi^{(1)}_2+qp\rho_1^{(1)}\rho_2^{(1)}\Psi^{(0)}_1  \Psi^{(0)}_2\bigg)\Bigg]\\
 &+2  \Psi^{(0)}  \Psi^{(1)} \Bigg[
 \sum_{i=1}^2\beta_i^2\left(-qp\rho\rho_i^{(1)}(\Psi^{(0)}_i)^2+(q-1)(p-1)\Psi^{(0)}_i\Psi^{(1)}_i\right)\\
 &\hskip 1cm+\beta_1 \beta_2
 \left(
 [-q  p \rho  \sum_{i=1}^2\rho_i^{(1)}+(q -1) (p-1) \rho_{12}^{(1)}]\Psi^{(0)}_1\Psi^{(0)}_2
 +(q-1)(p-1)(\Psi^{(0)}_1\Psi^{(1)}_2+\Psi^{(0)}_2\Psi^{(1)}_1)
 \right)
 \Bigg]\\
&-2( \Psi^{(0)}) ^3 \left[\beta_1 \beta_2 (p-1)  \rho_{12}^{(1)}   \Psi^{(1)}_{12}-\lambda  p
\sum_{i=1}^2(\beta_i \rho_i^{(1)} \Psi^{(1)}_i)\right]
+(q -1) (1-p)  (\Psi^{(1)})^2 \left(\sum_{i=1}^2\beta_i \Psi^{(0)}_i\right)^2, \label{eq:Phi}
\end{align}
\begin{align}
\Phi_{30}^{+}&=\frac{ q^2 \Gamma \rho \Psi^{(1)}}{ \Psi^{(0)} }  \(  \Sumi \beta_i  \Psi^{(0)}_i\) \(  \Sumi \beta_i  \rho_{i}^{(1)}\Psi^{(1)}_i\) + q \Psi^{(1)}  \(   \Gamma \lam \(  \Sumi \beta_i  \rho_{i}^{(1)}\Psi^{(1)}_i\)+ \beta_1 \beta_2 \rho_{12}^{(1)}   \Psi^{(1)}_{12}  \),\\
\Phi_{31}^{+}&= q \(-\Psi^{(1)} +    \frac{\lam^2 \Gamma}{q}  (T-t) \Psi^{(1)} - \Gamma q \rho^2 \Psi^{(1)} (T-t)   \frac{ \left(\Sumi \beta_i  \Psi^{(0)}_i \right)^2}{ {(\Psi^{(0)})^2 } }\)\\
&+\frac{q-1}{ \Psi^{(0)} }(T-t) \( \(\Sumi \beta_i  \Psi^{(0)}_i\) \left(\Sumi \beta_i  \Psi^{(1)}_i\right)  + \rho_{12}^{(1)} \beta_1\beta_2 \Psi^{(0)}_1\Psi^{(0)}_2\),\\
\Phi_{41}^{+}&=-(T-t) q \(\Gamma \(\lam + \frac{q\rho}{  \Psi^{(0)} }\(  \Sumi \beta_i  \Psi^{(0)}_i\) \) \(  \Sumi \beta_i  \rho_{i}^{(1)}\Psi^{(1)}_i\)    +\beta_1 \beta_2 \rho_{12}^{(1)} \sigma  \Psi^{(1)}_{12} \)\\
\Phi_{42}^{+}&=-(T-t) q \(1 - \frac{\Gamma\lam^2}{2q} +(T-t)q\rho^2\(\frac{\Sumi \beta_i  \Psi^{(0)}_i}{ \Psi^{(0)}}   \)^2  \).
\end{align}

From the equations \eqref{eq:Psi0} and  \eqref{eq:Psi1} satisfied by $\Psi^{(0)}$ and $\Psi^{(1)}$    respectively, the terms of order one and of order $\eps$ in \eqref{eq:v-} cancel.
For $p<0$, we have $\Gamma<0$ and consequently  $\frac1p\left[\frac{\Gamma}{2q}\lambda^2(T-t)-1\right]>0$. Therefore, from the boundedness of $\Phi_{20},\Phi_{30}^{+}, \Phi_{31}^{+}, \Phi_{41}^{+}, \Phi_{42}^{+}$, one can choose $M>0$ independently of $\eps$ such that the term in $\eps^2$ in \eqref{eq:v-} is positive. 
Finally, since the $\eps^3$ and $\eps^4$ terms are all bounded, it follows that for $\eps>0$ small enough the estimate \eqref{Qestimates} for $Q^{\pi^0}[v^{+}]$ follows.

Note that for deriving \eqref{subestimate} from this estimate, one needs to check that the martingale parts are true martingales. This can be seen by writing these quantities explicitly and using again  the boundedness of  the derivatives of $\Psi^{(0)}$ and $\Psi^{(1)}$ and the admissibility of $\pi^0$. We omit the details.


%

%
%

\subsubsection{Sub-solution, computation of $\sup_{\pi}Q^{\pi}[v^{-}]$}\label{sec:supersol}
Using the fact that $v^{-}_{xx}<0$, a similar calculation with any admissible strategy $\pi$ reveals: 

\begin{align}
&\frac{Q^{\pi}[v^{-}]}{{\( \Psi^{(0)}  + \eps \Psi^{(1)} -\eps^2M(T-t) \)^{q-2} }}
 \leq \frac{\sup_\pi Q^{\pi}[v^{-}]}{{\( \Psi^{(0)}  + \eps \Psi^{(1)} -\eps^2M(T-t) \)^{q-2} }}
 \\
 =& \left(\d_t v^{-} + \L_{z_1, z_2}^{\rho_{12}} v^{-} -\frac{\(\lam v^{-}_x + \Sumi \rho_i\beta_i v^{-}_i\)^2}{2 v^{-}_{xx}}\right)\( \Psi^{(0)}  + \eps \Psi^{(1)} -\eps^2M(T-t) \)^{2-q} \\
&=
q\(\frac{x^p}{p}\( \Psi^{(0)}  + \eps \Psi^{(1)} -\eps^2M(T-t)\)\)
 \(\d_t \Psi^{(0)}  +\frac \Gamma{2q} \lam^2\Psi^{(0)}+ \L ^{1,\rho}_{z_1,z_2}\Psi^{(0)}\) \\
&+ \eps q\(\frac{x^p}{p}\( \Psi^{(0)}  + \eps \Psi^{(1)}-\eps^2M(T-t) \)\) \(\d_t\Psi^{(1)} + \L^{1,\rho}_{z_1,z_2}  \Psi^{(1)} + \frac{\Gamma}{2q}\lam^2 \Psi^{(1)}  -  f_1(\Psi^{(0)}, \nabla \Psi^{(0)}, \mathbb H(\Psi^{(0)}) )\)\\
& 
 +\eps^2q\(\frac{x^p}{p}\Psi^{(0)}  \)
  \left(M\left[1-\frac{\Gamma}{2q}\lambda^2(T-t)\right]+
 \frac{\Phi_{20}+p\Theta^2}{2(1-p)  (\Psi^{(0)}) ^2}\right) \\
 & + \eps^3 q \(\Phi_{30}^{-}+ M\Phi_{31}^{-}\) + \eps^4 qM(T-t) \( \Phi_{41}^{+} + M\Phi_{42}^{-}\) ,\label{eq:v+}
\end{align}
where $\Phi_{20}$ is given by \eqref{eq:Phi} and
\begin{align}\label{def:Theta}
\Theta^2&=q^2\rho^2
\left(
\Psi^{(1)}\sum_{i=1}^2\beta_i\Psi^{(0)}_i-\Psi^{(0)}\sum_{i=1}^2\beta_i\Psi^{(1)}_i
\right)^2,
\end{align}
\begin{align}
\Phi_{30}^{-}&= \Gamma    \Sumi \beta_i  \rho_{i}^{(1)} \Psi^{(1)}_i\(\lam \Psi^{(1)} + \rho q \beta_i  \Psi^{(1)}_i\)+ \beta_1\beta_2 \( \Psi^{(1)} \Psi^{(1)}_{12} \rho_{12}^{(1)} - \( \rho_{12}^{(1)}(q-1) - \Gamma (\rho_{i}^{(1)} + \rho_{i}^{(1)} )q\)\Psi^{(1)}_{1}\Psi^{(1)}_{2}
 \),\\
\Phi_{31}^{-}&= \(\frac{\Gamma\lambda^2}{2q}(T-t)-1\)-\frac{(T-t)}{\Psi^{(0)}}\(  \Gamma q\rho\Sumi\ \(\beta_i\Psi^{(1)}_{i}\)^2  - \beta_1\beta_2\( \rho_{12}^{(1)}(q-1) - \Gamma (\rho_{i}^{(1)} + \rho_{i}^{(1)} )q\)\Psi^{(1)}_{1}\Psi^{(1)}_{2}
 \),\\
\Phi_{41}^{-}&=    \Gamma\lam\Sumi \beta_i  \rho_{i}^{(1)}\Psi^{(1)}_i   +\beta_1 \beta_2 \rho_{12}^{(1)} \sigma  \Psi^{(1)}_{12}, \\
\Phi_{42}^{-}&=- \(1 - \frac{\Gamma\lam^2}{2q}(T-t)  \).
\end{align}

Using the boundedness of $\Phi_{20},\Theta, \Phi_{30}^{-}, \Phi_{31}^{-}, \Phi_{41}^{-}, \Phi_{42}^{-}$, and the fact that $p<0$ one can choose $M>0$ independently of $\eps$ such that the term in $\eps^2$ in \eqref{eq:v+} is negative,
and  the other $\O(\eps^3),\O(\eps^4) $ {terms are} absorbed for $\eps$ small enough.

We conclude that the  inequality \eqref{Qestimates} for $\sup_\pi Q^{\pi}[v^-]$ holds. The martingale terms in 
\eqref{superestimatepi} are handled as before, before taking the supremum in the admissible $\pi$.

 Finally, from \eqref{subestimate} and \eqref{superestimatepi}, we deduce
 \begin{align}
\left|\E_{t,x,z_1,z_2}\[ \frac{1}{p}\left(X_T^{(\pi^{0}) }\right)^p\]- \frac{x^{p}}{p} \( \Psi^{(0)} + \eps \Psi^{(1)} \)^q \right|=x^p\O(\eps^2),
\end{align}
uniformly in $(t,z_1,z_2)$. {Note that here, the $\O(\eps^2)$ term depends on $M$.}
The ``near-optimality" estimate \eqref{pi0optimal} for the strategy $\pi^0$ follows easily from \eqref{approx}.

 \end{proof}
 
 \begin{remark} 
 Note that, as it should be, the additional term of oder $\eps^2$ from \eqref{eq:v-} to \eqref{eq:v+}
 \begin{align}
 \eps^2q\left\{\frac{x^p}{p}\( \Psi^{(0)}  \)^{q-1}\right\}\frac{p\Theta^2}{2(1-p)  (\Psi^{(0)}) ^2}  
\end{align}
is positive as $p$ simplifies, $q>0$, and $1-p>0$.
\end{remark}

\subsubsection{The case $0<p<1$}\label{sec:pnegative}

The conclusion of Theorem \ref{TheTheorem} holds modulo the following adjustments.

The proof in the case $0<p<1$ needs a different definition of $v^\pm$ because in that case $\Gamma=\frac{p}{1-p} >0$ and, therefore, the quantity $\left[1-\frac{\Gamma}{2q}\lambda^2(T-t)\right]$ may change sign. We redefine them as 
\begin{align}
v^\pm (t,x, z_1, z_2) = \frac{x^{p}}{p} \( \Psi^{(0)} (t,z_1, z_2) + \eps \Psi^{(1)} (t,z_1, z_2)\pm \eps^2M(-t)\)^q ,
\label{eq:sub-super-pnegative}
\end{align}
so that $v^+<v^-$ since $p<0$. 
The inequalities \eqref{Qestimates} still hold as we have now replaced $-\left[1-\frac{\Gamma}{2q}\lambda^2(T-t)\right]$ by $\left[1-\frac{\Gamma}{2q}\lambda^2(-t)\right]$ and $\Gamma>0$.
Now, we need to pay attention at terminal values. 
\begin{align}
v^\pm(T,x,z_1,z_2)&=\frac{x^{p}}{p} \( 1\pm \eps^2M(-T)\)^q,
\end{align}
so that $v^+(T,x,z_1,z_2)<\frac{x^{p}}{p}$ and $v^-(T,x,z_1,z_2)>\frac{x^{p}}{p}$. Then, the first line of \eqref{subestimate}
is replaced by 
\begin{align}
v(t,x,z_1,z_2)\geq \E_{t,x,z_1,z_2}\left[\U(X^{(\pi^0)}(T))\right]&\geq\E_{t,x,z_1,z_2}\left[v^+(T,X^{(\pi^0)}(T),Z_1(T),Z_2(T))\right],
\end{align}
and the first line of \eqref{superestimatepi} is replaced by 
\begin{align}
 \E_{t,x,z_1,z_2}\left[\U(X^{(\pi)}(T))\right]&\leq \E_{t,x,z_1,z_2}\left[v^-(T,X^{(\pi)}(T),Z_1(T),Z_2(T))\right].
\end{align}
The rest of the proof follows the same lines as in the case $p<0$.

\section{An Example with Explicit Formula}\label{sec:explicit1}
In our approach, the solution of non-linear HJB equation \eqref{eq:HJBfull} is approximated by $\frac{x^p}{p}\left(\Psi^{(0)}+\eps\Psi^{(1)}\right)$ where $\Psi^{(0)}$ and $\Psi^{(1)}$ are the solutions of the linear equations \eqref{eq:Psi0} and \eqref{eq:Psi1} respectively. The advantage is that these two equations are linear and also are of lower dimension being independent of $x$. In this section, we provide an example with explicit formulas for the approximation which we use as benchmark in a numerical illustration presented in Section \ref{sec:numeric}.

We consider the following model 
\begin{align}
\mu(z_1, z_2) = \bar\mu, \quad \sigma(z_1, z_2) =  \frac{\bar\sigma}{\sqrt{\eta_1\abs{z_1-z_2}+\eta_2\abs{z_1+z_2}+1}},\quad\bar \lam = \frac{\bar \mu}{\bar \sig},\label{eq:ex1.1}\\
\beta_i(z_1, z_2)  = \bar\beta
\sqrt{\abs{z_1+z_2}},\quad \alpha_i(z_i) = (m_1+m_2) + m_i \sign{z_1-z_2}  + z_i -\beta_i\lam\rho_i\Gamma,~i=1,2,
\label{eq:ex1.2}
\end{align}
with $\sign{x} = \ind_{\{x>0\}} - \ind_{\{x<0\}},$ and where recall that $\lam(z_1, z_2) = \frac{\mu}{\sig}(z_1, z_2) = \bar \lam \sqrt{\eta_1\abs{z_1-z_2}+\eta_2\abs{z_1+z_2}+ 1},$ 
Assume also that $m_1> m_2\ge1$ and $\eta_1, \eta_2\ge0$, and $(z_1,z_2)\in\K$, where $\K \subset \{(x_1, x_2)\in\R^2 \vert x_1>x_2,~ x_1+x_2>0\}$ is a compact. Note the singularity when $Z_1 =Z_2$. When $\eps=0$ and the  Brownian motions are perfectly correlated this does not happen since in that case $Z_1(s) > Z_2(s),~t\le s\le T$, as $d( Z_1 -Z_2)(s) = \((m_1 - m_2) \sign{  ( Z_1 -Z_2)(s)} + ( Z_1 -Z_2)(s) \)ds,$ and the boundary $z_1=z_2$ is absorbing. 
Therefore, it is not surprising that when $\eps>0$ small enough we should still be able to ignore the possibility of crossing the boundary with high probability.

Consider the case, $p<0$, the other case $p>0$ is similar. 
Next, observe that $Z_1+Z_2$  satisfies 
\begin{align}
d(Z_1 + Z_2)(s) = \((m_1+m_2) (2+\sign{(Z_1-Z_2)(s))} + (Z_1 + Z_2)(s)\)ds + \bar\beta\sqrt{2(1+\rho_{12}) \abs{(Z_1 + Z_2)(s)} }d\bar B_s,
\end{align}
where $\bar B$ is a one-dimensional Brownian motion, defined by $d\bar B_s = d\frac{\bar B_1(s) + \bar B_2(s)}{\sqrt{2(1+\rho_{12} )} }$, where in turn the Brownian motions $\bar B_i(s)$ are defined similar to \eqref{eq:drift-change} as $d\bar B_i(s) =  d B_i(s) - \Gamma \lam\rho_i ds.$
Compare $Z_1+Z_2$ with the following two diffusions
\begin{align}
d \overline Z(s) =  \(3(m_1+m_2) + \overline Z(s)\)ds + \bar\beta\sqrt{2(1+\rho_{12}) \abs{\overline Z(s)} }d\bar B_s,~\overline Z(t) = z_1+z_2\\
d \underline Z(s) =  \(m_1+m_2 + \underline Z(s)\)ds + \bar\beta\sqrt{2(1+\rho_{12}) \abs{\underline Z(s)} }d\bar B_s,~\underline Z(t) = z_1+z_2.\label{eq:4Feller}
\end{align}
We have that $\underline Z, \overline Z$ are both CIR processes, and under our assumptions they both satisfy the Feller condition and therefore $\underline Z, \overline Z>0$ a.s.. Additionally, they also sandwich $Z_1 + Z_2$, i.e. $\underline Z\le Z_1 + Z_2\le \overline Z$. Therefore, $Z_1 + Z_2>0$ a.s., and all the absolute values of $Z_1+Z_2$ inside the square roots above, can be removed and written simply as $Z_1+Z_2.$ 
Next, note that a CIR model has a good rate function \cite{chiarini2014large} and therefore, for $\eps>0$ small enough, $\sqrt{\overline Z(s)} \le \frac1{\sqrt[4]{\eps}}$
for all $s\in[t,T]$ on a set $A_\eps$ with probability at least $1-\eps^7.$ 
Therefore, the same is also true for $\sqrt{(Z_1+Z_2)(s)} \le \frac1{\sqrt[4]{\eps}}$ there.

Next, we have that 
\begin{align}
d( Z_1 -Z_2)(s) = \((m_1 - m_2) \sign{  ( Z_1 -Z_2)(s)} + ( Z_1 -Z_2)(s) \)ds + \bar\beta\sqrt{2(1-\rho_{12}) {(Z_1 + Z_2)(s)} }d\hat B_s ,
\label{eq:4Feller1}
\end{align}
where $\hat B$ is another one-dimensional Brownian motion. 
Recall that $z_1-z_2>0$. Using the fact that $1-\rho_{12}=O(\eps)$, we have that $ \bar\beta\sqrt{2(1-\rho_{12} ){(Z_1 + Z_2)(s)} } \le O(\sqrt[4]{\eps})$ on $A_\eps$, then, for $\eps>0$ small enough, we can further assume that on a set $B_\eps$ with probability at least $1-\eps^6$ the {process} $Z_1(s)-Z_2(s)\ge0$ on the entire $[t,T].$ 

Indeed, observe that on $[t,\tau\wedge T]$, where $\tau = \inf\{s>t \colon (Z_1-Z_2)(s) \le0\},$ we have that $(Z_1-Z_2)(s) =  (m_1-m_2) (\e{s-t}-1)  +\e{s} \( (z_1-z_2) \e{-t} + \int _t^s \e{-u} \bar\beta\sqrt{2(1-\rho_{12}) {{(Z_1 + Z_2)(u)} }}d\hat B_u\), ~t\le s \le \tau\wedge T.$
Note that $\check B_s = \int _t^{(\hat S)^{-1}(s)} \e{-u} \bar\beta\sqrt{2(1-\rho_{12}) {{(Z_1 + Z_2)(u)}} }d\hat B_u,$
is a Brownian motion on $[0,\hat S(T)]$, where $\hat S(u) = 2\int_t^{ u} \e{-2 \xi } \bar\beta^2 {(1-\rho_{12}) {{(Z_1 + Z_2)(\xi)}} }d\xi,$ and $(\hat S)^{-1}$ is its inverse. Let $\check\tau = \inf\{s>0 \colon \check B_s \le z_2 -z_1\},$ and recall that $m_1-m_2>0.$   
Therefore using the fact that $\hat S(T) = O\(\sqrt{\eps}\)$, on $A_\eps$
it can be then calculated that for $\eps>0$ small enough, $\P(\{\tau <T\}\cap A_\eps) \le \P(\{\check \tau< \hat S(T)\}\cap A_\eps) \le \eps^7. $

On $B_\eps$, we can get rid all the absolute values in \eqref{eq:ex1.1}-\eqref{eq:ex1.2}.
Thus, as opposed to finding the true PDE solution $\Psi$ in \eqref{eq:HJB-non-lin}\, we will instead proceed to find  the solution to the approximate PDE
\begin{align}
&  \d_t \tilde \Psi + {\bar\beta^2}\frac{z_1+z_2}{2}\(\tilde \Psi_{11} + 2\rho_{12} \Psi_{12} + \Psi_{22}  \) +\frac {\Gamma\bar\lam^2}{2q} (\eta_1(z_1-z_2)+\eta_2(z_1+z_2)+ 1 )\tilde  \Psi + \Sumi \(m_1+m_2+m_i+z_i\) \tilde\Psi_i  \\
&+ {\bar\beta^2} (z_1+z_2) \(\frac12 \Sumi \( (q-1)+q \rho_i^2\Gamma\) \frac{ \(\tilde\Psi_i \)^2}{\tilde\Psi } + \(\Gamma q \rho_1\rho_2 + (q-1)\rho_{12}\) \frac{ \tilde\Psi_1 \tilde\Psi_{2}  }{\tilde \Psi }\)=0, \mbox{ for } 0\le t<T, \label{eq:HJB-mod}\\
&0<z_1+z_2<\eps^{-1/4},~ 0<z_1-z_2,\\
&\tilde \Psi (T, z_1, z_2) = 1,~\tilde \Psi (t, z_1, z_2) =0, \mbox{ on } z_1-z_2 = 0,~z_1 +z_2 = \eps^{-1/4}.\label{eq:bnd-mod}
\end{align}
More specifically, consider the problem 
\begin{align}\label{eq:v-mod}
\tilde v(t,x,z_1, z_2) = \sup_{\pi} \E_{t,x,z_1, z_2}\[\U(X_T) \ind_{\{ Z_1-Z_2 >0,~ Z_1 +Z_2 \le \eps^{-\frac14} \}}\].
\end{align}
Here we need to tweak the definition of admissibility and additionally require that  $\E\[   \U^{2} (X_T^{\pi})\]<\infty$ in order for strategy $\pi$ to be admissible. Below we will show that $\pi^0$ is admissible.
Then
$v \le \tilde v\le 0,$ 
for $(z_1, z_2)\in\K$, we have that 
\begin{align}
\abs{v - \tilde v} (t,1, z_1, z_2) &\le  \E_{t,1,z_1, z_2}\[\abs{\U(X_T^{\pi^*})} \ind_{\{ Z_1-Z_2 < 0\} \cup \{ Z_1 +Z_2 > \eps^{-\frac14} \}}\] 
\le \sqrt{ \E_{t,1,z_1, z_2}\[\U^{2}(X_T^{\pi^*})\] \P(B_\eps^c)}
\le O(\eps^3).
\end{align}
Therefore 
\begin{align}
\abs{\Psi - \tilde \Psi} (t,z_1, z_2)\le  O(\eps^3),~(z_1, z_2)\in\K.\label{eq:Psi-tildePsi}
\end{align}

The modified problem \eqref{eq:v-mod} leads to the HJB equation \eqref{eq:HJB-mod} with the boundary conditions \eqref{eq:bnd-mod} with standard transformation. Therefore, we proceed to solve \eqref{eq:HJB-mod}, \eqref{eq:bnd-mod}. By the above discussion, we can also ignore the boundary conditions. Performing the asymptotic expansion leads to the same equations \eqref{eq:Psi0} and \eqref{eq:Psi1}.  
Therefore, for convenience, we will drop the tilde, and still call the $\eps$-approximations to $\tilde \Psi$ by $\Psi^{(0)},\Psi^{(1)}$ and the zero order associated strategy $\pi^{0}$.

Therefore, \eqref{eq:Psi0}, the equation satisfied by $\Psi^{(0)}$ becomes:
\begin{align}
&\d_t \Psi^{(0)}  + \sum_{i=1}^2 (m_1+m_2 + m_i  +z_i)  \Psi^{(0)}_i + {\bar\beta^2}\frac{z_1+z_2}{2} \( \Psi^{(0)}_{11} + 2 \Psi^{(0)}_{12} + \Psi^{(0)}_{22} \)  \label{eq:Psi0-ex}\\
&\quad+\frac \Gamma{2q} \bar\lam^2(1+(\eta_1+\eta_2)z_1-(\eta_1-\eta_2) z_2) \Psi^{(0)}
=0,\\
&\Psi^{(0)}(T, z_1, z_2)=1.
\end{align}
We use the standard ansatz $\Psi^{(0)} = \e{A(t) + \bar B_1(t) z_1 + \bar B_2(t)z_2} .$ Then
\begin{align}
\bar B_1'(t)&=\frac{\bar\lam^2\Gamma}{ 2q}(\eta_1+\eta_2) + \bar\beta^2 \frac12 (\bar B_1(t) + \bar B_2(t))^2 + \bar B_1(t),~\bar B_1(T)=0,\\
\bar B_2'(t)&=-\frac{\bar\lam^2\Gamma}{2 q} (\eta_1-\eta_2)+ \bar\beta^2\frac12 (\bar B_1(t) + \bar B_2(t))^2 + \bar B_2(t),~\bar B_2(T)=0,\\
A'(t) & = (2m_1+m_2) \bar B_1(t) + (m_1+2m_2) \bar B_2(t)+\frac{\bar\lam^2\Gamma}{2 q},~A(T)=0.
\end{align}
Letting $\hat B_1(t) = \bar B_1(t) + \bar B_2(t),~\hat B_2(t) = \bar B_1(t) -\bar B_2(t),$ then
\begin{align}
\hat B_1'(t)&= \bar\beta^2\hat B_1^2(t) + \hat B_1(t) +\frac{\bar\lam^2\Gamma}{ q}\eta_2 ,~\hat B_1(T) =0,\\
\hat B_2'(t)&= \hat B_2(t) + \frac{\bar\lam^2\Gamma}{ q}\eta_1~\hat B_2(T) =0.
\end{align}
The solution is given by $\hat B_2(t) =  \eta_1\frac{\bar\lam^2\Gamma}{ q} \(\e{- (T-t)} -1\), $ and
$\hat B_1(t) = a_{+}a_{-} \frac{1-\e{\bar\beta^2(T-t)(a_{+}-a_{-})  } }{a_{+}-a_{-}\e{\bar\beta^2(T-t)(a_{+}-a_{-}) }   }$, where $a_{\pm}=\frac{-1\pm \sqrt{1 - \bar\beta^2 \frac{4\bar\lam^2\Gamma \eta_2}{ q} }}{2{\bar\beta^2}},$ are assumed to be two distinct real roots of the quadratic $\bar\beta^2 a^2 +a +  \frac{\bar\lam^2\Gamma}{ q}\eta_2= 0.$ The latter is achieved, for example, if $p<0.$ 

Therefore, 
\begin{align}
\bar B_1(t) &= \frac{a_{+}a_{-}}{2} \frac{1-\e{\bar\beta^2(T-t)(a_{+}-a_{-})  } }{a_{+}-a_{-}\e{\bar\beta^2(T-t)(a_{+}-a_{-}) } }  +  \frac{\bar\lam^2\Gamma\eta_1}{ 2q} \(\e{- (T-t)} -1\),\\
\bar B_2(t) &= \frac{a_{+}a_{-}}{2} \frac{1-\e{\bar\beta^2(T-t)(a_{+}-a_{-})  } }{a_{+}-a_{-}\e{\bar\beta^2(T-t)(a_{+}-a_{-}) } }  -  \frac{\bar\lam^2\Gamma\eta_1}{ 2q} \(\e{- (T-t)} -1\),
\end{align}
and 
\begin{align}
A(t) &=\frac32 (m_1+m_2)  \((t-T) a_{-}-\frac1{\bar\beta^2}\log \frac{a_{-}\e{\beta^2(t-T) (a_{-}-a_{+}) } -a_{+}}{a_{-}-a_{+}} \) +  \frac{\bar\lam^2\Gamma\eta_1}{2q}(m_1-m_2) (1-\e{-(T-t)}) \\
&+ \frac{\bar\lam^2\Gamma}{ 2q}(1-m_1-m_2)(T-t).
\end{align}
Additionally, from \eqref{eq:pi0} the strategy $\pi^{0}$ is given by:
\begin{align}
\pi^{0} =\frac{x}{(1-p)\sigma(z_1, z_2)}\left( \lambda(z_1, z_2) + \rho q\, \Sumi \beta_i(z_1, z_2) \bar B_i(t)  \right).
\label{eq:pi0-ex}
\end{align}
We note that $\pi^{0}$ is an admissible strategy. Let $\bar \pi_0 =\frac{\pi^0}{x} =\frac{1}{(1-p)\sigma(z_1, z_2)}\left( \lambda(z_1, z_2) + \rho q\, \Sumi \beta_i(z_1, z_2) \bar B_i(t)  \right).$  From the fact that 
$\abs{ \bar B_i(t)} \le \frac{\bar\lam^2\abs{\Gamma}(\eta_1+\eta_2)}{ 2q} \(1-\e{- (T-t)} \) +O(\bar\beta^2),$ we conclude that
$\abs{\pi^0}$ is bounded by $\frac{\bar\lam}{\bar\sig(1-p)} (\nu_1\abs{z_1-z_2} + \nu_2(z_1+z_2) +1 + O(\bar\beta))x$.  Since $X_T = x\e{\int_t^T \mu \bar \pi^0_s - \frac{\sig^2(Z_1(s), Z_2(s)) (\pi^0_s)^2}{2} dt + \int_t^T \sig(Z_1(s), Z_2(s))\pi^0_s dW_s  } ,$ we get that
\begin{align}
&\E\[\U^2\( X_T\)\]  \\
&\le C\U^2(x)\sqrt[3]{\E\[\U^6\( \e{ \int_t^T (6\bar\mu + 18\bar\sig^2)\frac{\bar\lam}{\bar\sig(1-p)} \nu_1  \abs{Z_1-Z_2}(s) ds} \)\] } \sqrt[3]{\E\[\U^6\( \e{ \int_0^T (6\bar\mu + 18\bar\sig^2)\frac{\bar\lam}{\bar\sig(1-p)} \nu_2 (Z_1+Z_2)(s) ds} \)\]} . 
\end{align}
A technical, but simple calculation via affine ansatz solution then shows that 
for $\bar\beta>0$ small enough, such that
$1-8\bar\beta^2  (6\bar\mu + 18\bar\sig^2)>0,$ we have that
$\E\[\U^6\( \e{ \int_0^T (6\bar\mu + 18\bar\sig^2)\frac{\bar\lam}{\bar\sig(1-p)} \nu_2 (Z_1+Z_2)(s) ds} \)\] <\infty$ is finite.  Similar calculation can be done with other term involving $\abs{Z_1(s) -Z_2(s)}$ by utilizing the fact that 
$\abs{Z_1(s) - Z_2(s)} \le \check Z(s),$
where
\begin{align}
d \check Z(s) =  \(m_1-m_2 + \check Z(s)\)ds + \bar\beta\sqrt{2(1-\rho_{12}) \abs{\check Z(s)} }d\bar B_s,~\check Z(t) = z_1+z_2.
\end{align}
%
%
%
%

We want to highlight, that this is really the strategy $\tilde \pi^{0}$, i.e. the ``nearly-optimal" strategy associated with the zero order expansion of $\tilde \Psi$, but since the difference between $\tilde \Psi$ and $\Psi$ is small, this strategy will also achieve the desired accuracy level of $O(\eps^2).$ Additionally, while it is possible to repeat this entire calculation in the other two cases, when $z_1=z_2$ and $z_1<z_2$, and find the solution and the zero order strategy, this is not necessary in order to find a ``nearly-optimal" strategy. We start with $z_1>z_2,~z_1+z_2>0$, and then choose $\eps>0$ small enough, such that $(z_1, z_2)\in B_\eps$. While the process $(Z_1(t), Z_2(t))$ will leave the set $B_\eps$ with strictly positive probability, since this probability is very small, as explained above, this can be ignored. In other words, we can employ the strategy $\tilde{\tilde \pi}^{0} = \frac{x}{(1-p)\sigma(z_1, z_2)}\left( \lambda(z_1, z_2) + \rho q\, \Sumi \beta_i(z_1, z_2) \bar B_i(t)  \right) \ind_{\{ z_1-z_2 > 0\} \cup \{ z_1 +z_2 < \eps^{-1/4} \}}$, and it will still be ``nearly-optimal", as implied by \eqref{eq:Psi-tildePsi}.

Moreover, in the case $\eta_2=0,$ we can also find the $O(\eps)$  term.
Indeed, the PDE satisfied by $\Psi^{(1)}$ is
\begin{align}
&\d_t \Psi^{(1)}  + \sum_{i=1}^2 (m_1+m_2 + m_i  +z_i)  \Psi^{(1)}_i + \bar\beta^2\frac{z_1+z_2}{2} \( \Psi^{(1)}_{11} + 2 \Psi^{(1)}_{12} + \Psi^{(1)}_{22} \)  +\frac \Gamma{2q} \bar\lam^2(1+\eta_1(z_1- z_2)) \Psi^{(1)}\label{eq:Psi1-ex}\\
&\quad=\hat f_1(t,t,z_1- z_2, z_1+z_2),\\
&\Psi^{(1)}(T, z_1, z_2)=0.
\end{align}
where
\begin{align}
\hat f_1 (t,s,x,y) &= \Psi^{(0)}  \(  y \bar\beta^2\bar B_2(s)\bar B_1(s)\( \rho  q\Gamma (\tr_{1}-\rho  \tr_{12}+\tr_{2}) +\tr_{12} \)   + y  {q\Gamma \rho}\bar\beta^2 \Sumi  \bar B_i^2(s) \tr_{i}    \right.\\
&\left.\qquad+\bar\beta\Gamma\bar \lam\sqrt{y\( 1+ \eta_1 \(m_2-m_1 + \e{ s-t} ( x + m_1-m_2)\) \) } \Sumi \tr_i \bar B_i(s) \).
\end{align}
From the Feynman--Kac representation \eqref{eq:Psi1-E}  for $\Psi^{(1)}$,  we have that
\begin{align}
\Psi^{(1)}(t, z_1, z_2) 
&= \int_t^T \Exp{ \frac{\Gamma\bar\lam^2}{2q}   \(  \eta_1(m_1-m_2 + z_1 -z_2) \( \e{s-t}  -1\)+(1+ \eta_1 (m_1-m_2)) (s-t)\) } \\
&\quad\times\int_\R \hat f_1\(t,s, z_1-z_2,(\e{s-t} -1)  y\) \psi(y)dy ds .\label{eq:Psi1-ex-FK}
%
\end{align}
where we have used that $(Z_1 +Z_2)(s)$ evolves as a CIR process under $\tilde B$, the Brownian motion given by \eqref{eq:drift-change}:
$d(Z_1 + Z_2)(s) = 2(m_1+m_2) + (Z_1 + Z_2)(s) + 2 \bar\beta\sqrt{(Z_1 + Z_2)(s)}d\tilde B_s,$ and $(\e{s-t} -1) (Z_1 + Z_2)(s),~s\ge t$ has the p.d.f. 
\begin{align}
\psi(y) &= \frac1{2\(\e{s-t} -1\)} \e{-\frac1{2(\e{s-t} -1)} \((z_1+z_2) \e{s-t} + y\)   } \(\frac{y}{(z_1+z_2) \e{s-t}} \)^{(m_1+m_2)/2} \\
&\times I_{m_1+m_2}\(2\frac1{2(\e{s-t} -1)}  \sqrt{(z_1+z_2) \e{s-t}  y} \),
\end{align}
where $I_{m_1+m_2} (\cdot)$ is the modified Bessel function of the first kind of order $(m_1+m_2)$.
%
%

We note, that this calculation can also be done for the other two cases $z_1<z_2$ and $z_1=z_2$. 
\begin{remark}
 The model used in this example is based on square-root processes and does not satisfy the assumptions of Theorem \ref{TheTheorem}. Extending the accuracy result to that case requires another stopping argument at the first time one of the two processes $Z_1+Z_2$ or $Z_1-Z_2$  exit the interval $[\delta,\delta^{-1}]$ 
for some small parameter $\delta>0$. The stopped model satisfies the assumption but doesn't anymore allow  for explicit formulas for the functions $\Psi^{(0}$ and $\Psi^{(1)}$. A careful argument is needed to pass to the limit $\delta\to 0$ uniformly in $\eps$. This was done, for instance, for another nonlinear perturbation problem in \cite{FouqueNing} in the context of stochastic volatility uncertainty. It is quite technical and beyond the scope of this paper.
 \end{remark}

\subsection{Numerical Illustration}\label{sec:numeric}
We  illustrate our finding in the previous Section \ref{sec:explicit1} numerically.
We use the parameters:
\begin{align}
&\bar \beta = 0.3,~\bar \mu = 0.05,~ m_1 = 2,~ m_2 = 1,~ \bar \sig = 0.2,~  \bar\lam = \frac{\bar \mu}{\bar \sig} = \frac{0.05}{0.2},~ \rho_{12} = 1-\eps,~ \rho_1 = 0.5+\eps,~ \rho_2 =0.5 + \eps,\\ 
&\eta_1=1,~ \eta_2=0,~ p=-1,~ T=1,~ \eps=0.1.
\label{eq:param}
\end{align}
The graphs are all drawn as functions of $(0,z_1,z_2),~z_1>z_2\ge0,~ z_1+z_2, z_1-z_2\ge0, $ at the point $t=0.$ In this case it is easily seen that the Feller condition for the diffusions $\underline Z, \overline Z$ in \eqref{eq:4Feller} is satisfied. 
Figure \ref{fig1} illustrates: $\Psi - \Psi^{(0)}$ -- the difference between the numerical solution of $\Psi$, and its $O(1)$ approximation $\Psi^{(0)}$ (top left);
$\Psi - \Psi^{(\pi^0)}$ -- the difference between the numerical solution of $\Psi$, and the numerical solution of the non-linear HJB equation using the strategy  $\pi^0$ from \eqref{eq:pi0-ex}  (top right); 
and $\Psi -\( \Psi^{(0)} + \eps \Psi^{(1)}\)$ -- the difference between $\Psi$ and its $O(\eps)$ approximation $\Psi^{(0)} + \eps \Psi^{(1)}$(bottom).
%
First, note that we expect from Theorem \ref{TheTheorem} and the approximation \eqref{eq:Psi-tildePsi} that these errors are of order $\O(\eps)$, $\O(\eps^2)$ and  $\O(\eps^2)$ respectively. 
Next, it can be computed that $\hat f_1>0$ for $z_1+z_2>0$. Since we approximately have that $\Psi - \Psi^{(0)} \approx \eps \Psi^{(1)}. $ Therefore, we observe as expected that $ \Psi - \Psi^{(0)}$ is positive. Finally, since $v- \E_{t,x,z_1, z_2}\[\U\(X_T^{\pi^0}\)\]>0$, and $p<0$, it follows that $\Psi - \Psi^{(\pi^0)}<0$, which is again consistent with the sign observed in Figure \ref{fig1}. We want to emphasize that as expected this graph shows that if the simple zero order approximating strategy $\pi^{0}$ is used the difference in utility is of order $O(\eps^2)$, and thus as expected this is a "nearly-optimal" strategy.

\begin{figure}[ht!]
\begin{center}
\includegraphics[scale=0.6]{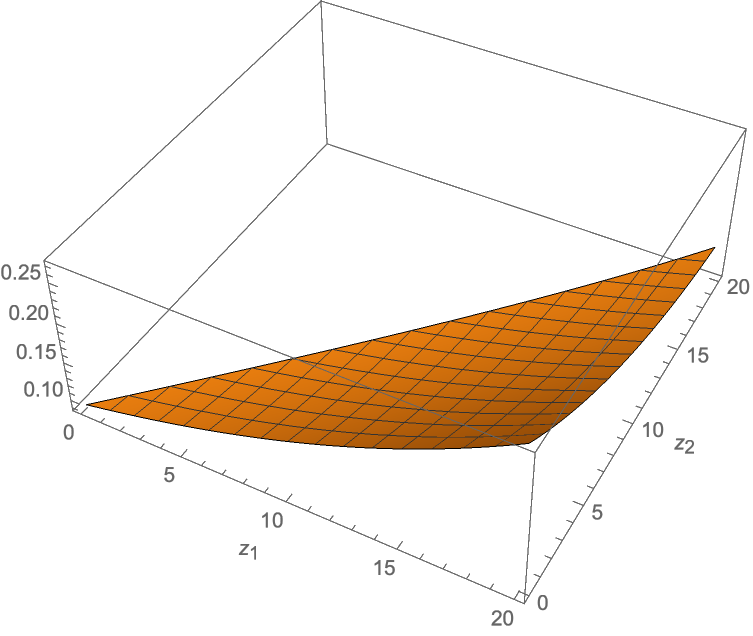}
\hspace{20pt}
\includegraphics[scale=0.6]{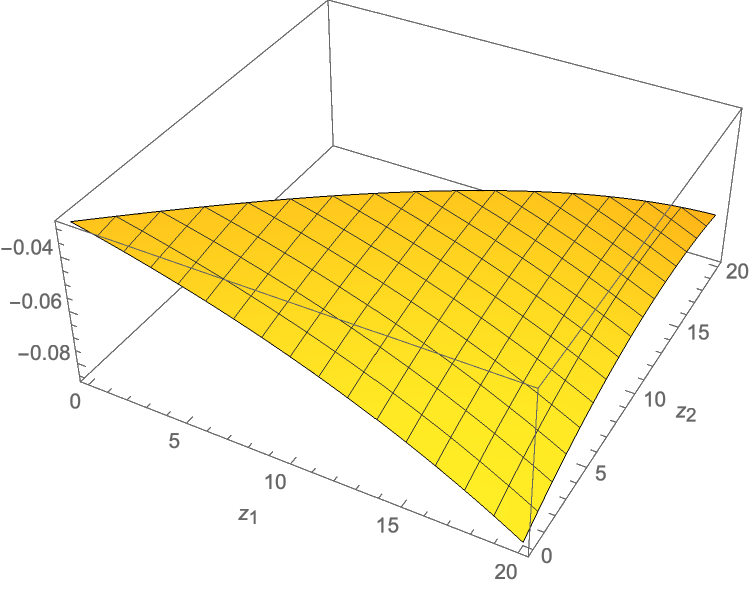}
\includegraphics[scale=0.6]{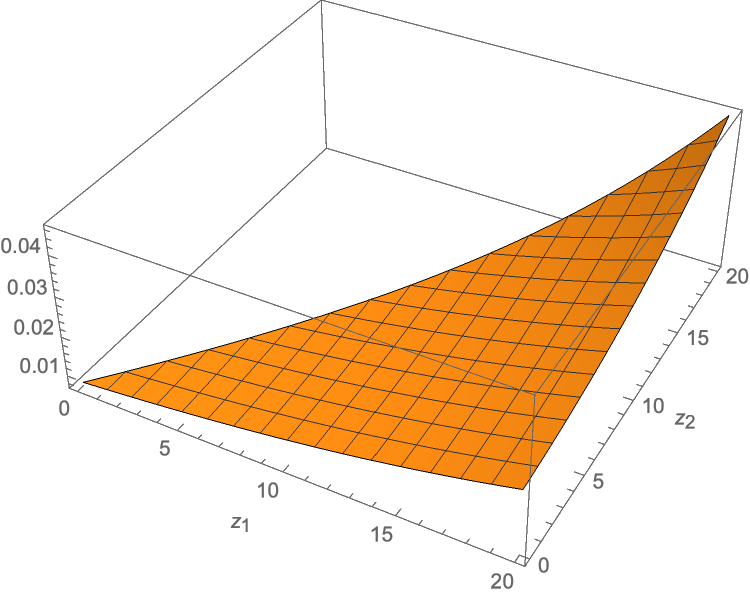}
\end{center}

\caption{Top left: graph of $\Psi - \Psi^{(0)}$ -- the difference  between the numerical solution $\Psi$, and its $O(1)$ approximation $\Psi^{(0)}$. 
Top right:  the graph of $\Psi - \Psi^{(\pi^0)}$ -- the difference between the numerical solution of $\Psi$, and the numerical solution of the HJB equation but using the strategy $\pi^0$.
Bottom: the graph of $\Psi -\( \Psi^{(0)} + \eps \Psi^{(1)}\)$ -- the difference between the numerical solution $\Psi$, and its $O(\eps)$ approximation $ \Psi^{(0)} + \eps \Psi^{(1)}.$ 
All graphs are done as a function of $(z_1,z_2),~z_1>z_2\ge0,~ z_1+z_2, z_1-z_2\ge0, $ at the point $t=0.$  }
\label{fig1} 
\end{figure}

\section{Extension to Models with Multi Assets}
We now show how to extend our results to a model with multi-assets, and multi-factors. Consider a model with multiple assets governed by
\begin{align}
\frac{dS_i(t)}{S_i(t)} &= \mu_i(\mathbf{Z}(t))dt + \sigma_i(\mathbf{Z}(t))\,dW_i(t),~i=1,2,\cdots,n,  \label{eq:S-multiD} \\
dZ_k(t) &=	 \alpha_k(Z_k(t))\,dt +\beta_k( Z_k(t))\,dB_k(t), ~k=1,2,\cdots,m, \label{eq:Z-multiD} 
\end{align}
where we use the vector notation $\mathbf{Z} := (Z_1,  \ldots, Z_m)^T$ and the
correlation structure between the Brownian motions $(W_1,\cdots,W_n,B_1,\cdots,B_m)$ is given by:
 $$
 d\<W_i, W_j\>_t =\rho^W_{ij} dt, \quad d\<B_k, B_l\>_t =\rho^B_{kl} dt,\quad d\<W_i, B_k\>_t =\rho_{ik}dt,
 \quad 1\leq i,j\leq n,\, 1\leq k,l\leq m,
 $$
 with parameters $(\rho^W_{ij},\rho^B_{kl},\rho_{ik} )$ ensuring a proper correlation structure (in particular $\rho^W_{ii}=\rho^B_{kk}=1$ and symmetries $\rho^W_{ij}=\rho^W_{ji}$, $\rho^B_{kl}=\rho^B_{lk}$).
 
 Assuming that the wealth is fully invested in the $n$ stocks in a self-financed way, then  the wealth process is given by:
\begin{align}
d X(t) = \sum_{i=1}^n\pi_i(t) \frac{dS_i(t)}{S_i(t)} =  \sum_{i=1}^n\pi_i(t) \left[ \mu_i(\mathbf{Z}(t))dt + \sigma_i(\mathbf{Z}(t))\,dW_i(t)\right],
\label{eq:wealth-multi}
\end{align}
where $\pi_i(t)$ is the amount invested in asset $i$  at time $t$.
The value function of the optimal investment problem with terminal time $T$ and utility $\U$ is:
\begin{align}
v(t,x,\mathbf{z}) = \sup_{\mathbf{\pi}} \E_{t,x,\mathbf{z}}\[\U(X_T)\],\quad \U(x) = \frac{x^{p}}{p},~p<1, p\ne0.
\end{align}
We define the following operators:
\begin{align}
\L^{\pi,\rho^W,\rho}_{x,\mathbf{z}} &= \sum_{i=1}^n\mu_i(\mathbf{z})\pi_i \frac{\d}{\d x}+ \frac12\( \sum_{i,j=1}^n \rho^W_{ij}\pi_i\pi_j \sig_i(\mathbf{z})\sig_j(\mathbf{z})   \)\frac{\d^2}{\d x^2} +  \sum_{i=1,k=1}^{n,m}\rho_{ik}\pi_i\sig_i(\mathbf{z})\beta_k(z_k) \frac{\d^2}{\d x \d z_k},\\
\L^{\rho^B}_{\mathbf{z}} &=\sum_{k=1}^m \alpha_k(z_k) \frac{\d}{\d z_k} 
+\frac12\( \sum_{k,l=1}^m\rho^B_{kl}\beta_k(z_k)\beta_l(z_l)
 \frac{\d^2}{\d z_k \d z_l}\).
\end{align}
The value function $v$ then satisfies:
\begin{align}
&\d_t v + \L^{\rho^B}_{\mathbf{z}} v + \sup_{\pi} \L^{\pi,\rho^W,\rho}_{x,\mathbf{z}} v =0,\label{eq:HJB-multi}\\
& v(T, x, \mathbf{z}) = \U(x).
\end{align}

Our asymptotics will be around the case where the  Brownian motions $B_k$ are fully correlated. In order to model this regime, we define:
\begin{align}\label{rhoeps}
\rho_{ik} = \rho_i +  \rho_{ik}^{(1)}\eps, ~1\leq i\leq n,\, 1\leq k\leq m,\quad \mbox{ and }\quad \rho_{kl}^B = 1 +   \trB_{kl}\eps,
\end{align}
with $\trB_{kk}=0$ and $\trB_{kl}<0$, and
$\eps$ is a small parameter,  $0<\eps\ll 1$, small enough to ensure a proper covariance structure.

\begin{remark}
The model that we are perturbing corresponding to $\eps=0$ in \eqref{rhoeps}, cannot be of {\it eigenvalue equality (EVE)} type as considered in \cite{EVE2018} unless $m=1$, that is models with a single factor. Indeed, the matrix $\rho\rho^T$ with $\rho_{ik} = \rho_i , ~1\leq i\leq n,\, 1\leq k\leq m$, admits zero as eigenvalue as soon as $m\geq 2$ and therefore, cannot satisfy the EVE condition $\rho\rho^T=cI$ unless in the uncorrelated case $\rho=0$.
\end{remark}

In order to keep the formulas as explicit as possible, we present the case with two assets and two factors.

\subsection{Model with Two Assets}\label{sec:twoassets}
We continue illustrate the calculation of the expansions in an example with two assets and two driving factors. Therefore the model will now be governed by \eqref{eq:S-multiD}--\eqref{eq:Z-multiD} with $n=m=2$.
Maximization over $\pi$ in \eqref{eq:HJB-multi} gives:
\begin{align}
\pi^{*}_i = &\frac{\sig_j \(\sig_i \Sumk \beta_k\rho_{ik}v_{xk} +\mu_i v_x\)-\rho^{W}_{12} \sig_i \(\sig_j \Sumk \beta_k\rho_{jk}v_{xk} +\mu_jv_x\)}{\left((\rho^{W}_{12})^2-1\right) \sig_i^2 \sig_j v_{xx}}, i,j=1,2, ~i\ne j,
\label{eq:pi-multi}
\end{align}
where $v_{k}$ denotes a derivative with respect to $z_k,~k=1,2$, and $(\rho^{W}_{12})^2<1$ to ensure that the two stocks are not fully correlated.
Substituting \eqref{eq:pi-multi} into \eqref{eq:HJB-multi}, it follows that 
\begin{align}
&\d_t v + \L^{\rho^B_{12}}_{\z} v \label{eq:HJB-multi1}\\
&-\frac{ \Sumi \beta_i^2 \left(-2 \rho_{2i} \rho_{1i} \rho^{W}_{12}+\rho_{1i}^2+\rho_{2i}^2\right) v_{x, i}^2+2 \beta_1 \beta_2 \left(\rho_{21} \left(\rho_{22}-\rho_{12} \rho^{W}_{12}\right)+\rho_{11} \left(\rho_{12}-\rho_{22} \rho^{W}_{12}\right)\right) v_{x,1} v_{x,2}}{2(1-(\rho^{W}_{12})^2) v_{xx}}\\
&-\left(\Sumi \beta_i v_{x,i} \left(\rho_{1i} (\lam_1 -\lam_2 \rho^{W}_{12} )+\rho_{2i} (\lam_2 -\lam_1 \rho^{W}_{12} )\right)\right)\frac{ v_x }{(1-(\rho^{W}_{12})^2) v_{xx}}-\frac{\left(\lam_1^2 -2 \lam_1 \lam_2 \rho^{W}_{12} +\lam_2^2 \right) v_x^2}{2(1-(\rho^{W}_{12})^2) v_{xx}}=0,
\end{align}
where $\lam_i(z_1, z_2) = \frac{\mu_i(z_1, z_2)}{\sigma_i(z_1, z_2)},~i=1,2.$

For $q\in\R$, we again perform a distortion transformation \eqref{eq:distortion} of the HJB equation \eqref{eq:HJB-multi1} for the value  function $v$. Similar to \eqref{eq:HJB-non-lin}, $\Psi$ must satisfy:
%
%
\begin{align}
&\d_t\Psi+\frac12\(\beta_1^2 \Psi_{11} + 2\rho^B_{12} \beta_1\beta_2 \Psi_{12} + \beta_2^2 \Psi_{22}\) +  \frac{\Gamma \left(\lam_1^2 -2 \lam_1 \lam_2 \rho^W_{12} +\lam_2^2\right)}{2q(1-(\rho^W_{12})^2)} \Psi\\
&+\Sumi \(\alpha_i +  \frac{\Gamma\beta_i}{1-(\rho^W_{12})^2} \left(\rho_{1i} (\lam_1 -\lam_2 \rho^{W}_{12} )+\rho_{2i} (\lam_2 -\lam_1 \rho^{W}_{12} )\right)  \)\Psi_i\\
%
&+\frac12\sum_{i,j=1,2,~ j\ne i}  \beta_i^2 \left( q-1  +q\frac{\Gamma}{1-(\rho^W_{12})^2}  \left(\rho _{ii}^2-2 \rho _{ji} \rho _{ii} \rho^W_{12}+\rho _{ji}^2\right)\right) \frac{(\Psi_i)^2}{\Psi}\\
&+ \beta_1\beta_2\( q \frac{\Gamma}{1-(\rho^W_{12})^2} \(\rho_{22}\rho_{21}+\rho_{11}\rho_{12} - \rho_{12}^W(\rho_{12}\rho_{21}+\rho_{22}\rho_{11}) \) +(q-1)\rho^B_{12}\)\frac{\Psi_1\Psi_2}{\Psi}=0.
%
\end{align}

\subsubsection{Fully Correlated Case}
Analogous to Section \ref{sec:full-corr}, we temporarily assume  that  the two stochastic volatility factors are fully correlated:
 $\rho_{12}^B=1$, $B_1(t) = B_2(t) = B(t)$, and $d\<W_i, B\>_t =\rho_{i} dt,~i,j=1,2.$

 We consider the ansatz
\begin{align}
v(t,x, z_1, z_2) = \frac{x^{p}}{p} \(\Psi^{(0)}(t,z_1, z_2)\)^q.
\end{align}
Let $\Gamma = \frac{p}{1-p}.$ Then, it follows that $\Psi^{(0)}$ satisfies

\begin{align}
&\d_t\Psi^{(0)}+\frac12\(\beta_1^2 \Psi^{(0)}_{11} +  \beta_1\beta_2 \Psi^{(0)}_{12} +\frac12 \beta_2^2 \Psi^{(0)}_{22}\) +  \(\alpha_1 +  \frac{\Gamma\beta_1}{1-(\rho^W_{12})^2}\(\lam_1\(\rho_1 - \rho_2 \rho^W_{12}\) + \lam_2\(\rho_2 - \rho_1 \rho^W_{12}\)   \)   \)\Psi^{(0)}_1\\
&+ \(\alpha_2 +  \frac{\Gamma\beta_2 }{1-(\rho^W_{12})^2}\(\lam_1\(\rho_1 - \rho_2 \rho^W_{12}\) + \lam_2\(\rho_2 - \rho_1 \rho^W_{12}\)  \)  \)\Psi^{(0)}_2+\frac{\Gamma \left(\lam_1^2 -2 \lam_1 \lam_2 (\rho^W_{12}) +\lam_2^2\right)}{2q(1-(\rho^W_{12})^2)} \Psi^{(0)}\\
&+\frac12\left( q-1  +q\frac{\Gamma}{1-(\rho^W_{12})^2}  \left(\rho_{1}^2-2 \rho_{1} \rho_{2} \rho^W_{12}+\rho_{2}^2\right) \right) \frac{\(\beta_1\Psi^{(0)}_1+\beta_2\Psi^{(0)}_2\)^2}{ \Psi^{(0)}}=0.
\end{align}
Choosing $$q= \left(1+\Gamma \frac{\rho_1^2+ \rho_2^2 -2 \rho_1\rho_2 \rho_{12}^W}{1-(\rho_{12}^W)^2}\right)^{-1}$$ linearizes the equation to get:
\begin{align}
\d_t \Psi^{(0)}  +\frac{\Gamma}{2q(1-(\rho^W_{12})^2)} \left(\lam_1^2 -2 \lam_1 \lam_2 (\rho^W_{12}) +\lam_2^2\right)  \Psi^{(0)}+ \L \Psi^{(0)}
&=0,\label{eq:Psi0-multi}\\
\Psi^{(0)}(T, z_1, z_2)&=1,
\end{align}
where
\begin{align}
\L \phi =  \L_{\z}^{1} \phi +   \frac{\Gamma}{1-(\rho^W_{12})^2}\(\lam_1\(\rho_1 - \rho_2 \rho^W_{12}\) + \lam_2\(\rho_2 - \rho_1 \rho^W_{12}\)\) \Sumi \beta_i(z_i)\phi_i.
\end{align}
We have the Feynman--Kac representation:
\begin{align}
\Psi^{(0)}(t, z_1, z_2) = \tilde \E_{t,z_1,z_2}\[\e{\frac \Gamma{2q} \int_t^T \lam^2(Z_1(s), Z_2(s))ds } \],\label{eq:Psi0-E-multi} 
\end{align}
where $\tilde \P$ is defined so that $$\tilde B_t  = B_t - \frac{\Gamma}{1-(\rho^W_{12})^2}\int_0^t \(\lam_1(Z_1(s), Z_2(s))\(\rho_1 - \rho_2 \rho^W_{12}\) + \lam_2(Z_1(s), Z_2(s))\(\rho_2 - \rho_1 \rho^W_{12}\)\)ds$$ is standard Brownian motion under it, and we denoted
\begin{align}
\lambda^2=\frac{\lambda_1^2+\lambda_2^2-2\lambda_1\lambda_2\rho_{12}^W}{1-(\rho_{12}^W)^2}.
\end{align}

\subsubsection{Asymptotics}
In the general case, we will assume a correlation structure of the form \eqref{rhoeps}:
\begin{align}\label{rhoeps2}
\rho_{ik} = \rho_i +  \rho_{ik}^{(1)}\eps, ~i,k=1,2,\quad \mbox{ and }\quad \rho_{12}^B = 1 +   \trB_{12}\eps,
\end{align}
with $\trB_{12}<0$ and
$\eps$ is a small parameter,  $0<\eps\ll 1$, small enough to ensure a proper covariance structure. As was done previously in the case with one stock and a risk-free asset, we will now expand the general case, around the known case of $\eps=0$, and calculate the asymptotic expansion similar to \eqref{eq:ansatz-stoch-vol}. 
\begin{align}
v(t,x, z_1, z_2) = \frac{x^{p}}{p} \( \Psi^{(0)} (t,z_1, z_2)+\eps\Psi^{(1)} (t,z_1, z_2) +\O(\eps^{2})\)^q 
,
\label{eq:ansatz-stoch-vol-multi}
\end{align}
Note that the expansion has the same number of arguments as before, as there are still two factors, though the functions $\Psi^{(i)},~i=0,1,$ will be different.

Expanding the correlation coefficients as in \eqref{rhoeps2} and the value function as in \eqref{eq:ansatz-stoch-vol-multi}, 
 we see that $\Psi^{(1)}$ satisfies an equation similar to \eqref{eq:Psi1}:
\begin{align}
&\d_t\Psi^{(1)} + \L  \Psi^{(1)} +\frac{\Gamma}{2q(1-(\rho^W_{12})^2)} \left(\lam_1^2 -2 \lam_1 \lam_2 (\rho^W_{12}) +\lam_2^2\right) \Psi^{(1)}  +  f_1(\Psi^{(0)}, \nabla \Psi^{(0)}, \mathbb H(\Psi^{(0)}) )=0,
\label{eq:Psi1-multi}
\end{align}
where $\nabla \Psi^{(0)}, \mathbb H(\Psi^{(0)})$ denote the gradient and the Hessian of $\Psi^{(0)}$ and
\begin{align}
&f_1(\Psi^{(0)}, \nabla \Psi^{(0)}, \mathbb H(\Psi^{(0)}) )=\frac{q\Gamma }{(1-(\rho^W_{12})^2) \Psi^{(0)}}  \left(\Sumi \beta_i^2 \(  (\rho_1 - \rho_2\rho^W_{12}) \tr_{1i}+(\rho_2 - \rho_1\rho^W_{12}) \tr_{2i}\) (\Psi^{(0)}_i)^2\right.\\
&\left.+\beta_1 \beta_2 \Psi^{(0)}_1 \Psi^{(0)}_2
\(   (\rho_1 - \rho_2\rho^W_{12}) (\tr_{11} + \tr_{12})+(\rho_2 - \rho_1\rho^W_{12}) (\tr_{21} + \tr_{22})-(\rho_1^2 + \rho_2^2 - 2\rho_1\rho_2\rho^W_{12})\trB_{12}\)
\right)\\
& \frac{\Gamma}{1- (\rho^W_{12})^2} \Sumi \((\lam_1-\lam_2\rho^W_{12}) \tr_{1i}  +(\lam_2-\lam_1\rho^W_{12}) \tr_{2i} \) \beta_i  \Psi^{(0)}_i - \trB_{12}\beta_1\beta_2\Psi^{(0)}_{12}.
\end{align}

We now consider $\pi^0$, the first order approximation  to $\pi^{*}$ given in \eqref{eq:pi-multi}, by substituting the first order approximation for $v$ from \eqref{eq:ansatz-stoch-vol}, namely, $v(t,x, z_1, z_2) \approx \frac{x^{p}}{p} \(\Psi^{(0)}(t,z_1, z_2)\)^q.$

Therefore,
\begin{align}
\pi^{0}_i &= \frac{x \left(q   \left(\beta_1 \left(\rho_{i1}-\rho_{j1} \rho^{W}_{12}\right) \Psi_1^{(0)}+\beta_2 \left(\rho_{i2}-\rho_{j2} \rho^{W}_{12}\right) \Psi_2^{(0)}\right)+\Psi^{(0)}  (\lam_i -\lam_j \rho^{W}_{12} )\right)}{(1-p) \left(1-(\rho^{W}_{12})^2\right) \sig_i \Psi^{(0)} }~ i,j=1, 2, ~i\ne j.
\label{eq:pi0-multi}
\end{align}
We next use $(\pi_1, \pi_2) = (\pi^{0}_1,\pi^{0}_2)$ in the supremum of  \eqref{eq:HJB-multi} together with the expansions \eqref{rhoeps2}, \eqref{eq:ansatz-stoch-vol-multi} and evaluate the equation, to get that: 
\begin{align}
&\frac{\d_t v + \L_{\z}^{\rho^B_{12}} v + \L^{\pi^{0},\rho^W,\rho}_{x,\mathbf{z}} v} {v} \\
&= q\( \Psi^{(0)} + \eps  \Psi^{(1)}\) \Bigg[
{\(\d_t \Psi^{(0)}  +\frac{\Gamma}{2q(1-(\rho^W_{12})^2)} \left(\lam_1^2 -2 \lam_1 \lam_2 (\rho^W_{12}) +\lam_2^2\right) \Psi^{(0)}+ \L \Psi^{(0)}\)} \\
&  +
{\(\d_t\Psi^{(1)} + \L  \Psi^{(1)} + \frac{\Gamma}{2q(1-(\rho^W_{12})^2)} \left(\lam_1^2 -2 \lam_1 \lam_2 (\rho^W_{12}) +\lam_2^2\right) \Psi^{(1)}  -  f_1(\Psi^{(0)}, \nabla \Psi^{(0)}, \mathbb H(\Psi^{(0)})\)}\Bigg]\eps+ \O(\eps^2)\\
&= \O(\eps^2),
\end{align}
where the last equality is obtained by cancelling the first two terms  using the equations \eqref{eq:Psi0-multi} and \eqref{eq:Psi1-multi} satisfied by $\Psi^{(0)}$ and $\Psi^{(1)}$ respectively.

To summarize, this formal computation shows that the strategy $(\pi_1, \pi_2) = (\pi^{0}_1,\pi^{0}_2)$ given by \eqref{eq:pi0-multi} generates the value $v$ given by \eqref{eq:HJB-multi} up to order $\eps$.

\subsubsection{Explicit Formulas}

We again consider a specific choice of a model, similar to the example in Section \ref{sec:explicit1}.
Namely, we change \eqref{eq:ex1.1} and \eqref{eq:ex1.2} to account for two stocks to be:
\begin{align}
&\mu_i(z_1, z_2) = \bar\mu_i, ~~ \sigma_i(z_1, z_2) =  \frac{\bar\sigma_i}{\sqrt{\eta_1\abs{z_1-z_2}+\eta_2\abs{z_1+z_2}+1}},~~\bar \lam_i = \frac{\bar \mu_i}{\bar \sig_i},~~\beta_i(z_1, z_2)  = 
\sqrt{\abs{z_1+z_2}},\label{eq:ex1.1-multi}\\
&\alpha_i(z_1, z_2) = (m_1+m_2) + m_i \sign{z_1-z_2}  + z_i -\frac{\Gamma\beta_i \(\rho_{1i} (\lam_1 -\lam_2 \rho^{W}_{12} )+\rho_{2i} (\lam_2 -\lam_1 \rho^{W}_{12} )\)}{1-(\rho^W_{12})^2},~i=1,2,
\label{eq:ex1.2-multi}
\end{align}
%
with $n=m=2$ in \eqref{eq:S-multiD} and \eqref{eq:Z-multiD}. 
Similar to Section \ref{sec:explicit1}, we have that outside of a set with small probability $Z_1 + Z_2$ is bounded as $(Z_1+Z_2)(s) \le \frac1{\sqrt{\eps}}$, and $Z_1(s)-Z_2(s)\ge0.$ Therefore similar to \eqref{eq:HJB-mod}, let $\tilde \Psi$
be a solution to the approximating PDE:
\begin{align}
&\d_t\tilde \Psi+\frac{z_1+z_2}{2}\(\tilde\Psi_{11} + 2\rho^B_{12} \tilde\Psi_{12} + \tilde\Psi_{22}\) +  \(\eta_1(z_1-z_2)+\eta_2(z_1+z_2)+1\) \frac{\Gamma \left(\bar\lam_1^2 -2 \bar\lam_1\bar \lam_2 \rho^W_{12} +\bar\lam_2^2\right)}{2q(1-(\rho^W_{12})^2)}\tilde \Psi\\
&+\Sumi \(m_1+m_2+m_i+z_i  \)\tilde\Psi_i\\
&+\frac{z_1+z2}2\sum_{i,j=1,2,~ j\ne i}  \beta_i^2 \left( q-1  +q\frac{\Gamma}{1-(\rho^W_{12})^2}  \left(\rho _{ii}^2-2 \rho _{ji} \rho _{ii} \rho^W_{12}+\rho _{ji}^2\right)\right) \frac{(\tilde\Psi_i)^2}{\tilde\Psi}\\
&+ (z_1+z_2)\( q \frac{\Gamma}{1-(\rho^W_{12})^2} \(\rho_{22}\rho_{21}+\rho_{11}\rho_{12} - \rho_{12}^W(\rho_{12}\rho_{21}+\rho_{22}\rho_{11}) \) +(q-1)\rho^B_{12}\)\frac{\tilde\Psi_1\tilde\Psi_2}{\tilde\Psi}=0, \mbox{ for } 0\le t<T,\\ 
&0<z_1+z_2<\eps^{-1/4},~ 0<z_1-z_2,\\
&\tilde \Psi (T, z_1, z_2) = 1,~\tilde \Psi (t, z_1, z_2) =0, \mbox{ on } z_1-z_2 = 0,~z_1 +z_2 = \eps^{-1/4}.\label{eq:bnd-mod-multi}
\end{align}
Then as before \eqref{eq:Psi-tildePsi} holds. Therefore we solve for our approximations $\Psi^{(0)}$, $\Psi^{(1)}$, similarly to Section \ref{sec:explicit1}.
Namely, $\Psi^{(0)}$ is the solution to PDE:
\begin{align}
&\d_t \Psi^{(0)}  + \sum_{i=1}^2 (m_1+m_2 + m_i  +z_i)  \Psi^{(0)}_i + \frac{z_1+z_2}{2} \( \Psi^{(0)}_{11} + 2 \Psi^{(0)}_{12} + \Psi^{(0)}_{22} \)  \label{eq:Psi0-ex-multi}\\
&\quad+\frac{\Gamma \left(\bar\lam_1^2 -2 \bar\lam_1\bar \lam_2 \rho^W_{12} +\bar\lam_2^2\right)}{2q(1-(\rho^W_{12})^2)}(1+(\eta_1+\eta_2)z_1-(\eta_1-\eta_2) z_2) \Psi^{(0)}
=0,\\
&\Psi^{(0)}(T, z_1, z_2)=1,
\end{align}
which is the same as \eqref{eq:Psi0-ex}, only with $\bar\lam^2 =\frac{ \left(\bar\lam_1^2 -2 \bar\lam_1\bar \lam_2 \rho^W_{12} +\bar\lam_2^2\right)} {(1-(\rho^W_{12})^2)}$,
and therefore $\Psi^{(0)} = \e{A(t) + \bar B_1(t) z_1 + \bar B_2(t)z_2} $ can be solved the same way as in Section \ref{sec:explicit1}. 
Moreover, analogously to \eqref{eq:pi0-ex} from \eqref{eq:pi0-multi} we have that 
\begin{align}
\pi^{0}_i &= \frac{x \left(q   \sum_{k=1}^2 \left(\beta_k(z_1,z_2) \left(\rho_{ik}-\rho_{jk} \rho^{W}_{12}\right) \bar B_k(t)\right)+ (\lam_i(z_1,z_2) -\lam_j(z_1,z_2) \rho^{W}_{12} )\right)}{(1-p) \left(1-(\rho^{W}_{12})^2\right) \sig_i(z_1,z_2)  }~ i,j=1, 2, ~i\ne j.
\end{align}
The same is true regarding $\Psi^{(1)}$ in case $\eta_2=0$, then it is the solution of the PDE
\begin{align}
&\d_t \Psi^{(1)}  + \sum_{i=1}^2 (m_1+m_2 + m_i  +z_i)  \Psi^{(1)}_i + \frac{z_1+z_2}{2} \( \Psi^{(1)}_{11} + 2 \Psi^{(1)}_{12} + \Psi^{(1)}_{22} \)  \\
&\quad+\frac{\Gamma \left(\bar\lam_1^2 -2 \bar\lam_1\bar \lam_2 \rho^W_{12} +\bar\lam_2^2\right)}{2q(1-(\rho^W_{12})^2)}(1+\eta_1(z_1- z_2)) \Psi^{(1)}=\hat f_1(t,t,z_1- z_2, z_1+z_2).\label{eq:Psi1-ex-multi}\\
&\Psi^{(1)}(T, z_1, z_2)=0,
\end{align}
with
\begin{align}
&\hat f_1 (t,s,x,y) =
\frac{yq\Gamma \Psi^{(0)}}{(1-(\rho^W_{12})^2) }  \left(\Sumi \bar B_i^2(s) \(  (\rho_1 - \rho_2\rho^W_{12}) \tr_{1i}+(\rho_2 - \rho_1\rho^W_{12}) \tr_{2i}\) \right.\\
&\left.+\bar B_1(s) \bar B_2(s) 
\(   (\rho_1 - \rho_2\rho^W_{12}) (\tr_{11} + \tr_{12})+(\rho_2 - \rho_1\rho^W_{12}) (\tr_{21} + \tr_{22})-(\rho_1^2 + \rho_2^2 - 2\rho_1\rho_2\rho^W_{12})\trB_{12}\)
\right)\\
& \frac{\Gamma\sqrt{y\( 1+ \eta_1 \(m_2-m_1 + \e{ s-t} ( x + m_1-m_2)\) \) }}{1- (\rho^W_{12})^2} \Sumi \((\bar \lam_1-\bar\lam_2\rho^W_{12}) \tr_{1i}  +(\bar\lam_2-\bar\lam_1\rho^W_{12}) \tr_{2i} \) \bar B_i(s)  \Psi^{(0)} \\
&- \trB_{12}\bar B_1(s) \bar B_2(s) \Psi^{(0)}.
\end{align}
Again, it solution is given by \eqref{eq:Psi1-ex-FK} with $\bar\lam^2 =\frac{ \left(\bar\lam_1^2 -2 \bar\lam_1\bar \lam_2 \rho^W_{12} +\bar\lam_2^2\right)} {(1-(\rho^W_{12})^2)}$.

\section{Conclusion}
The problem of portfolio optimization with power utilities when returns and volatilities are driven by a single factor can be linearized by using a classical distortion transformation. In this paper we proposed to treat this same problem in the presence of several factors. Our approach is to consider a perturbation around the case where the factors are fully correlated which can be linearized and amenable to simpler equations. We identify the leading order term for the value function corresponding to a Merton's portfolio and we characterize the first order correction as the solution to a linear equation. An example with explicit solutions is given to illustrate the quality of the approximation. Under a set of reasonable assumptions, we rigorously establish an accuracy result for this regular perturbation problem by using the construction of sub- and super-solutions to the fully nonlinear HJB equation characterizing the value function. In turn, we deduce that the leading order approximation of the optimal strategy generates the value function up to the first order of accuracy.

\subsection*{Acknowledgement}
The authors would like to thank Ruimeng Hu for her comments on an earlier version of the paper. The authors are also grateful to the two referees whose comments and suggestions helped a lot in improving the paper.

\small{\bibliography{references2}}
\end{document}